\documentclass[11pt]{amsart}
\usepackage{graphicx,amssymb,amsmath,amsthm}
\usepackage{enumerate}
\usepackage{dsfont}
\usepackage[colorlinks, citecolor=red]{hyperref}
\usepackage{mathrsfs}
\usepackage{epsfig}
\usepackage{lscape}
\usepackage{subfigure}
\usepackage{epstopdf}

\usepackage{caption}
\usepackage{algorithm}
\usepackage{algpseudocode}
\usepackage{multirow}

\usepackage{geometry}
\usepackage{amsfonts}
\usepackage{graphicx,cite}

\usepackage{paralist}
\usepackage{enumerate}
\usepackage{url}

\newcommand{\R}{{\mathbb R}}

\newtheorem{definition}{Definition}[section]
\newtheorem{corollary}[definition]{Corollary}

\newtheorem{theorem}[definition]{Theorem}
\newtheorem{lemma}[definition]{Lemma}

\newtheorem{remark}[definition]{Remark}

\newtheorem{example}[definition]{Example}

\def\bfx{{\bf x}}
\def\bfy{{\bf y}}
\def\bfb{{\bf b}}
\def\bfv{{\bf v}}
\def\bfs{{\bf s}}
\def\bfe{{\bf e}}
\def\bfb{{\bf b}}

\def\bfa{{\bf a}}
\def\bfu{{\bf u}}

\date{}

\begin{document}

\baselineskip 14pt
\bibliographystyle{plain}

\title[Graph Sparsification via UGA]{Graph Sparsification \\ by Universal Greedy Algorithms}
\author{Ming-Jun Lai}
\address{Department of Mathematics, University of Georgia, Athens, GA 30602. U.S.A.}
\email{mjlai@uga.edu}

\author{Jiaxin Xie}
\address{School of Mathematical Sciences, Beihang University, Beijing, 100191, China }
\email{xiejx@buaa.edu.cn}

\author{Zhiqiang Xu}
\address{LSEC, Inst.~Comp.~Math., Academy of
Mathematics and System Science,  Chinese Academy of Sciences, Beijing, 100091, China \newline
School of Mathematical Sciences, University of Chinese Academy of Sciences, Beijing 100049, China}
\email{xuzq@lsec.cc.ac.cn}

\begin{abstract}
Graph sparsification is to approximate an arbitrary graph by a sparse graph and is useful
in many applications, such as  simplification of social networks, least squares problems,
numerical solution of symmetric positive definite linear systems and etc.
In this paper, inspired by the well-known sparse signal recovery algorithm
called orthogonal matching pursuit (OMP), we introduce a deterministic, greedy edge selection
algorithm  called universal greedy approach (UGA)  for  graph sparsification.
For a general spectral sparsification problem, e.g., positive subset selection problem from a set of $m$ vectors
from $\mathbb{R}^n$, we propose a nonnegative UGA algorithm which needs $O(mn^2+ n^3/\epsilon^2)$ time to find a
$\frac{1+\epsilon/\beta}{1-\epsilon/\beta}$-spectral sparsifier with positive coefficients
with sparsity $\le\lceil\frac{n}{\epsilon^2}\rceil$, where $\beta$ is the ratio between the smallest length and largest length of the vectors.   The convergence of the nonnegative UGA algorithm will
be established. For the graph sparsification problem,
another UGA algorithm will be proposed which can output a
$\frac{1+O(\epsilon)}{1-O(\epsilon)}$-spectral sparsifier
with $\lceil\frac{n}{\epsilon^2}\rceil$ edges in $O(m+n^2/\epsilon^2)$ time from a graph with $m$ edges and
$n$ vertices under some mild assumptions.
This is a linear time algorithm in terms of the number of edges
that the community of graph sparsification is looking for.
 The best result in the literature to the knowledge of the authors
is the existence of a deterministic algorithm which is almost linear,
i.e. $O(m^{1+o(1)})$ for some $o(1)=O(\frac{(\log\log(m))^{2/3}}{\log^{1/3}(m)})$.
Finally, extensive experimental results, including applications to graph clustering and least squares regression,
show  the effectiveness of proposed approaches.
\end{abstract}

\maketitle

\section{Introduction}
Graph sparsification aims to find a sparse subgraph from a dense graph $G$ with $n$
vertices and $m$ edges (typically $m\gg n$) so that the sparsified subgraph can serve
as a proxy for $G$ in numerical computations for graph-based applications. In
\cite{BSS14},  Batson, Spielman and Srivastava showed that for any undirected graph
$G$ one can find a sparse graph (sparsifier) whose graph Laplacian matrix can well
preserve the spectrum of the original graph Laplacian matrix. Such a spectral graph
sparsification plays increasingly important roles in many applications areas in
mathematics and computer science \cite{Kyng2016,Li2018, SpielmanTeng2014}. A related
research, known as \emph{Laplacian Paradigm}, is illustrated as an emerging paradigm
for the design of scalable algorithms in recent years. We refer the reader to
\cite{Spielmanbook2019, Teng2016, Vishnoi2013} for excellent surveys on its
background and applications.

Mathematically, we can state the graph sparsification problem as follows.
Consider a  undirected and weighted graph $G=(V, E, {\bf w})$, where $V$ is a set of vertices,
$E$ is a set of edges, and ${\bf w}$ is a weight function that assigns a positive weight to each edge.
The Laplacian matrix of graph $G$ is defined by
$$
L_{G}=\sum\limits_{(u,v)\in E}w_{(u,v)}(\bfe_u-\bfe_v)(\bfe_u-\bfe_v)^\top ,
$$
where  $w_{(u,v)}\ge 0$ is the weight of edge $(u,v)$ and $\bfe_u\in\mathbb{R}^{|V|}$ is the characteristic vector of vertex
$u$ (with a $1$ on coordinated $u$ and zeros elsewhere). In other words, for any $\bfx\in \mathbb{R}^n$,
$$
\bfx^\top L_G\bfx =\sum\limits_{(u,v)\in E}w_{(u,v)}\big(\bfx(u)-\bfx(v)\big)^2 \ge 0.
$$
That is, $L_G$ is positive semidefinite.
\emph{Spectral graph sparsification} is the process of approximating the graph $G$ by a sparse (linear-sized) graph
 $H=(V,\tilde{E}, \tilde{\bf w})$ such that
\begin{equation}
\label{GS}
a \bfx^\top L_G \bfx \le \bfx^\top L_H \bfx \le b\bfx^\top L_G\bfx
\end{equation}
for all $\bfx\in \mathbb{R}^{|V|}$, where $b\geq a>0$. Setting $\kappa:=b/a$, $H$ is
called a $\kappa$-approximation of $G$ or a $\kappa$-sparsifier of $G$. Actually, if
we restrict the inequality in \eqref{GS} only for all $\bfx\in\{0,1\}^{|V|}$, one can
obtain the cut sparsification \cite{bk96}. Batson, Spielman and Srivastava
\cite[Theorem $1.1$]{BSS14} proved that for every weighted graph $G$ and every
$\epsilon \in (0,1)$ there exists a weighted graph $H$ with at most
$\lceil(n-1)/\epsilon^2\rceil$ edges  which  is an
$\frac{(1+\epsilon)^2}{(1-\epsilon)^2}$-approximation of $L_G$. More generally, let
$V=\{\bfv_1, \ldots, \bfv_m\}\subset \mathbb{R}^n$ be a collection of vectors with
$m\gg n$. We replace $L_G$ by
\begin{equation}
\label{isotropic}
B= \sum_{i=1}^m \bfv_i \bfv_i^\top,
\end{equation}
which is clearly positive semidefinite. In \cite[Theorem $1.2$]{BSS14}, the authors
proved that for any $\epsilon\in (0,1)$, these exists a  $\bfs=(s_1,
\ldots, s_m)\in \mathbb{R}^m_+$ with  $\|\bfs\|_0\le \lceil \mbox{rank}(B)/\epsilon^2\rceil$ such that
\begin{equation}
\label{app21}
(1- \epsilon)^2 B \preceq \sum_{i=1}^m s_i \bfv_i \bfv_i^\top  \preceq (1+\epsilon)^2  B,
\end{equation}
where $\mathbb{R}_{+}^m$ denotes the nonnegative orthant in $\mathbb{R}^m$ and   $\|\bfs\|_0$ stands
for the number of nonzero entries of vector $\bfs$.

Our aim is to  use the ideas of the well-known orthogonal matching  pursuit (OMP) to
study spectral sparsification problems. We first focus on the following problem:
\begin{equation}\label{app2}
\min_{\bfs\in \mathbb{R}^m, \bfs\geq 0}\|\bfs\|_0\quad \text{\rm s.t.}\quad
(1- \epsilon)^2 B \preceq \sum_{i=1}^m s_i \bfv_i \bfv_i^\top  \preceq (1+\epsilon)^2  B,
\end{equation}
where
 $
 B= \sum_{i=1}^m \bfv_i \bfv_i^\top
$ and $\{\bfv_1 \ldots, \bfv_m\}\subset \mathbb{R}^n$.

We shall call (\ref{app2}) {\em sparse positive subset selection } problem and will
study it first. For convenience, we shall start with the isotropic case, i.e., when
$B=I_n$ is the $n\times n$ identity matrix. Such decomposition appear in many areas
of mathematics and are also called isotropic sets, tight frames, John's
decompositions (when their mean is zero), etc. We shall also explain how to do with
general symmetric positive semidefinite matrix $B$. Mainly we shall provide an
algorithm called nonnegative UGA algorithm to find such a sparse subset with positive
coefficients and establish the convergence of the algorithm. Then we solve graph
sparsification problem  using a similar UGA algorithm and establish the convergence.

\subsection{Related work}
In recent years, spectral sparsification is a widely studied topic and
has applications to many areas
in mathematics and theoretical computer science. Recent work
includes \cite{ALO15,Bansal2019,BSS14,batson2013spectral, chufocs2018,Peng2019,Feng2016,Feng2018,youssef2017,Hermsdorff2019,Jambulapati2018,Kyng2016, KPPS2017, LS2018, LS15,dCSHS16,
SpielmanSrivastava2011, SpielmanTeng2011,SpielmanTeng2014,Wuchen2020,
 Zhang-zhao-feng2020, zhao2018}.

In the seminal paper \cite{SpielmanTeng2011},  Spielman and Teng first introduced the
notion of spectral sparsification and showed that  any undirected graph $G$ of $n$
vertices and $m$ edges has a spectral sparsifier with $O(n \log^cn)$ edges that can
be computed in $O(m \log^c(m))$ time for some positive constant $c$.
 In terms of the number of vertices  $n$ with assumption of the number of edges $m$ being $O(n^2)$,
 the
Spielman and Teng's algorithm  need $O(n^2 2^c\log^2(n))$ time. Some progresses had
been made by the work of Spielman and Srivastava \cite{SpielmanSrivastava2011} and
Lee and Sun \cite{LS2018}, who showed how to find the spectral sparsifiers with
$O(n\log n)$ and $O(n)$ edges, respectively, in $O(m \log^c(m))$ time. We remark that
all of this type of algorithms requires random sampling or random projection.

A celebrated result of Batson, Spielman and Srivastava \cite{BSS14} states that for any undirected
graph $G$ of $n$ vertices and $m$ edges and a given parameter $\epsilon\in(0,1)$, there exists
a spectral sparsifer with $O(n/\epsilon^2)$ edges. They further provided a polynomial time, deterministic algorithm
that in each iteration one edge is chosen deterministically to optimize the change of some
`barrier' potential functions, however such algorithm and subsequent algorithms by
\cite{zouzias2012} require $O(mn^3/\epsilon^2)$ and $O(mn^2/\epsilon^2+n^4/\epsilon^2)$
time respectively.  When a graph has $m=O(n^2)$ edges, the computational time is $O(n^5)$ or $O(n^4)$.
 Most recently, the work \cite{Peng2019} solves a big
unsettled problem by showing that there exist a deterministic algorithm to find the spectral
sparsification in almost-linear time, i.e. $O(m^{1+o(1)})$ with
$o(1)=O(\frac{(\log\log(m))^{2/3}}{\log^{1/3}(m)})$.  Again, in terms of the number of
vertices, the time is $O(n^{2+o(1)})$.
Table \ref{tablesum} is a summary of various algorithms for graph sparsification.

\begin{table}
\caption{Summary of various algorithms for graph sparsification with $n$ and $m$ being the number of
nodes and edges of the input graph $G$, respectively. D/R stands for deterministic or randomized algorithms.\label{tablesum}}
\centering
{\small
\begin{tabular}{ |c|  c| c| c| c|   }
\hline
\multicolumn{1}{ |c| }{Algorithms}  &\multicolumn{1}{c| }{Sparsifier Size} &\multicolumn{1}{c| }{Approximation } &\multicolumn{1}{ c| }{Flops Count } &\multicolumn{1}{c| }{D/R}\\
\hline
  Theorem $6.1$ in \cite{Jambulapati2018} \cite{SpielmanTeng2011}& $O(n \log^cn)$ & $\frac{1+\epsilon}{1-\epsilon}$-approx.&$O(m \log^c(m))$ & Random.  \\
\hline
Theorem $1.1$ in \cite{BSS14} & $O(n)$ & $\big(\frac{1+\epsilon}{1-\epsilon}\big)^2$-approx.&$O(mn^3/\epsilon^2)$ & Determ.   \\
\hline
Corollary $7.3$ in \cite{Peng2019} & $O(n\log^2 n)$ & $n^{o(1)}$-approx.&$O(m^{1+o(1)})$ & Determ.   \\
\hline
Theorem $1.1$ in \cite{LS2018} & $O(n)$ & $\frac{1+\epsilon}{1-\epsilon}$-approx.&$O(m \log^c(m))$ & Random.  \\
\hline
Theorem $1$ in \cite{SpielmanSrivastava2011} & $O(n\log n)$ & $\frac{1+\epsilon}{1-\epsilon}$-approx.&$O(m \log^c(m))$ & Random.  \\
\hline
Algorithm \ref{alg-eco}  & $O(n/\epsilon^2)$ & $\frac{1+O(\epsilon)}{1-O(\epsilon)}$-approx.& $O(m+n^2/\epsilon)$ & Determ.  \\
\hline
\end{tabular}
}
\end{table}

\subsection{Our motivation and contribution}
In this paper, we  propose a simple and efficient algorithm based on a universal
greedy approach (UGA) to solve these graph sparsification problems to have at most
$n/\epsilon^2$ vectors or edges for any given $\epsilon\in (0,1)$. Our motivation can
be explained as follows. Let us consider (\ref{app2}) as an example. Denote by
$\phi_i\in \R^{n^2}$ the vectorization of matrix $\bfv_i \bfv_i^\top$ for $i=1,
\ldots, m$, and let $\Phi =[\phi_1, \ldots, \phi_m]$ be the sensing matrix and denote
by $\bfb\in \R^{n^2}$ the vetorization of $B$. We can rewrite  (\ref{isotropic})  as
the following underdetermined linear system
$$
 \Phi \bfe = \bfb,
$$
 where $\bfe\in \mathbb{R}^m$ is the vector with $1$ for all entries.
We can formulate the problem \eqref{app2} as a compressive sensing \cite{D06,FR13} problem:
Find a sparse solution $\bfs\ge 0$ with $\|\bfs\|_0\ll m$ such that
\begin{equation}
\label{CS}
\min \|\bfs\|_0 \quad \mbox{s.t.}  \quad \|\Phi \bfs -\bfb\|_2\leq \epsilon, \bfs\ge 0.
\end{equation}
Thus, the sparsification problem in (\ref{app2}) is a constrained compressive sensing problem.

In the standard compressive sensing approach, i.e. the problem in (\ref{CS}) without
$\bfs\ge 0$ or the following version (\ref{CS0}), one efficient way to solve it is to
use a greedy algorithm, which is  called orthogonal matching pursuit (OMP)
\cite{omp1993}:
\begin{equation}
\label{CS0}
\min \|\Phi \bfs - \bfb\|_2^2 \quad \mbox{s.t.} \quad \|\bfs\|_0 \le s. 
\end{equation}
OMP is a very popular algorithm and has been studied by many researchers. See, e.g.,
\cite{CDD17,omp1993,T04,tropp-omp,XuOMP,Z11} and many variations of the OMP. We refer
to \cite{DT05} and \cite{FK14} for an approach for (\ref{CS}). Besides, the OMP
procedure  is also very efficient in completing a low rank matrix with given partial
known entries (see \cite{wang2014, wang2015}). To deal
with positive subset selection problem (\ref{app2}), we have to enforce the
nonnegativity in the OMP algorithm. To deal with the graph sparsification
problem (\ref{GS}), we have to improve the efficiency of the OMP algorithm. In this paper, we mainly extend the ideas in the
orthogonal rank $1$ matrix pursuit (OR1MP) algorithm in \cite{wang2014,wang2015},
especially, the  economical OR1MP to the settings of the computation  for the positive subset selection and graph
sparsification.

Our main results can be summarized as follows.
For the sparse positive subset selection problem, under some mild assumptions, our Algorithm~\ref{alg-ecoIso} will produce a subset $\{\bfv_{j_\ell}, \ell=1,
\cdots, s\}$ from (\ref{isotropic}) and positive coefficients $c(j_\ell)$ such that
\begin{equation}
\label{napp2}
(1- \epsilon/\beta) B \preceq \sum_{\ell=1}^s c(j_\ell) \bfv_{j_\ell} \bfv_{j_\ell}^\top  \preceq (1+\epsilon/\beta)  B,
\end{equation}
with $s\le \hbox{rank}(B)/\epsilon^2$, where $\beta$ denotes the ratio between the smallest length and largest length of the vectors given in \eqref{newformulae}.
Our computational time is $O(m n^2/\epsilon^2)$ which is faster than the one in \cite{BSS14} which needs $O(m n^3/\epsilon^2)$.

For the graph sparsification problem, under the some mild assumptions, our Algorithm~\ref{alg-eco} will output  a
sparsified graph $H_i$ with $n/\epsilon^2$ edges  in
$O(m+n^2/\epsilon^2)$ time  such that
\begin{equation}
\label{nGS}
(1- O(\epsilon)) L_G \preceq L_{H_i}  \preceq (1+O(\epsilon))  L_G.
\end{equation}
Our computational time $O(m+n^2/\epsilon^2)$ is linear in the number of edges.
The best result in the literature to our knowledge is
the recent study in \cite{Peng2019} as mentioned above, where the researchers show that
there is a deterministic algorithm for finding graph sparsification in
a  near-linear time $O(m^{1+o(1)})$ with $o(1) =O(\frac{(\log\log(m))^{2/3}}{\log^{1/3}(m)})$.

In addition, the sparse subset selection problem can provide a linear sketching method for solving least squares problem
according to literature \cite{boutsidis2013l2,yosef2018,woodruff2014}. The main ideas of linear sketching are to
compress the data matrix $A$ and  observation vector $\bfb$ as small as possible before doing linear regression.  Once we
use Algorithm~\ref{alg-ecoIso} to find a linear sketching $\tilde{A}$ and $\tilde{{\bf b}}$, we propose a new method to solve the
linear regression  instead of the linear sketching method. We shall provide a theorem and numerical results to justify  our new
method.

\subsection{Notation and organization}
Let $A=(a_1,\ldots,a_n)\in\mathbb{R}^{n\times n}$ be an $n\times n$ real matrix.  We
use $A(i,j)$ to denote the $(i,j)$-th component of $A$. The operator norm and the
Frobenius norm of $A$ is defined as $\|A\|_2$ and
$\|A\|_F:=\sqrt{\sum_{i,j}A(i,j)^2}$, respectively. For any $\bfx\in\mathbb{R}^n$, we
use $\|\bfx\|_2$ to denote the $\ell_2$-norm of $\bfx$. For any subset $\Lambda$, we
use $A_{\Lambda}$ to denote the sub-matrix of $A$ obtained by extracting the columns
of $A$ indexed by $\Lambda$. Let $vec(A):=(a_1^\top ,\ldots,a_n^\top )^\top $ denote
a vector reshaped from matrix $A$ by concatenating all its column vectors. The inner
product of two matrices $A$ and $B$ is defined as $\langle A,B\rangle:=\langle
vec(A),vec(B)\rangle$. We call a matrix $A$ positive semidefinite if $\bfx^\top A
\bfx\geq 0$ holds for any $\bfx\in\mathbb{R}^n$, and a matrix $A$ positive definite
if $\bfx^\top A \bfx>0$ holds for any nonzero $\bfx\in\mathbb{R}^n$. For any two
matrices $A$ and $B$, we write $A\preceq B$ to represent $B-A$ is positive
semidefinite, and $A\prec B$ to represent $B-A$ is positive definite.
Finally, for any
$b_1,b_2\in\mathbb{R}$, we often use $b_1=O(b_2)$ if $|b_1/b_2|$ is bounded from the
above.

The paper is organized as follows. In the next section, we shall first discuss the
problem (\ref{app2}) to produce a subset $S\subset \{1, \cdots, m\}$ with nonnegative
coefficients $s_i>0, i\in S$ and establish the convergence of the nonnegative UGA
algorithm together with computational complexity.  Then we will study the problem of
(\ref{GS}) in Section 3 to show an universal greedy approach (UGA) algorithm will find a
desired graph sparsifier.
 We shall begin with the computational complexity of Algorithm~\ref{alg-eco} and then establish
the convergence to justify that the sparsified graph satisfies (\ref{nGS}).
 In the end of this paper, we shall present some numerical results on graph sparsification and least
squares regression, where we show that our new method is more accurate than the standard linear sketching method.

\section{The nonnegative UGA algorithm for isotropic subsets}
In this section, we shall propose an universal greedy approach (UGA) to find a sparse positive subset
 $\bfs=(s_1, \ldots, s_m)\in \mathbb{R}^m_+$ with  $\|\bfs\|_0\le \lceil
n/\epsilon^2\rceil$ and $V_{\bf s}=\{\bfv_{s_i}, i=1, \cdots, |\bfs|\}$ from the given set $V
=\{\bfv_1, \cdots, \bfv_m\}$ satisfying $\displaystyle \sum_{i=1}^m \bfv_i \bfv_i^\top=I_n$ such that
(\ref{app2}) holds with $B$ being replaced by $I_n$.
Our UGA strategy follows from the ideas of the economical OR1MP  in \cite{wang2015} which will significantly
speed up the computation as stated in Theorem~\ref{complexity3} below.
As this technique can be used to speed up all OMP like algorithms,  we call it
the universal greedy approach (UGA) instead.

\begin{algorithm}[H]
\caption{The nonnegative UGA for subset selection \label{alg-ecoIso}}
\begin{algorithmic}
\Require
$V=\{\bfv_1,\bfv_2,\ldots,\bfv_m\}\subset \mathbb{R}^{n}$ with $\sum\limits_{i=1}^m\bfv_i\bfv_i^{\top}=I_n$, and $\epsilon\in(0,1)$.
\begin{enumerate}
\item[1:] $R_0:=I_n$, $\Lambda_0:=\emptyset$, $L_0:=0$, $c_0:=0\in\mathbb{R}^m$ and $i:=1$.
\item[2:] Find the index $j_i$ such that
\begin{equation}
\label{positive}
 \bfv_{j_i}^\top R_{i-1}\bfv_{j_i} =\max_{\bfv\in V} \bfv^\top R_{i-1}\bfv,
\end{equation}
and update $\Lambda_i=\Lambda_{i-1}\cup\{j_i\}$.
\item[3:] Compute the optimal weights
$$(\alpha_1^i,\alpha_2^i)=\arg\min\limits_{(\alpha_1,\alpha_2)\in\mathbb{R}^2}
\big\|I_n-\alpha_1L_{i-1}-\alpha_2\bfv_{j_i}\bfv_{j_i}^{\top}\big\|^2_F.
$$
\item[4:] Update
$$
L_{i}=\alpha_1^iL_{i-1}+\alpha^i_2\bfv_{j_i}\bfv_{j_i}^{\top}
$$
and
$$
R_i=I_n-L_{i}.
$$
\item[5:] Update the coefficient $c_{i}=\alpha_1^ic_{i-1}$ and $ c_{i}(j_i)=\alpha_2^i+c_{i-1}(j_i)$.
\item[6:] If $i>\lceil\frac{n}{\epsilon^2}\rceil$, stop and go to output. Otherwise, set $i=i+1$ and return to Step $2$.
\end{enumerate}

\Ensure
The sparsifier $L_{\lceil\frac{n}{\epsilon^2}\rceil}$, the selected index $\Lambda=\Lambda_i$  and the coefficient $c=c_i$.
\end{algorithmic}
\end{algorithm}

Mainly, we update Step $2$ by using (\ref{positive}) instead of
$$
 |\bfv_{j_i}^\top R_{i-1}\bfv_{j_i}| =\max_{\bfv\in V} |\bfv^\top R_{i-1}\bfv|
$$
in Algorithm~\ref{alg-eco}.
This will ensure the non-negativity of the coefficients of subsets in all iterations. See Section \ref{sec3.2} after the
discussion of the computational complexity.

\subsection{Computational complexity}
We shall first establish  the following result for the  running time of Algorithm~\ref{alg-ecoIso}.
\begin{theorem}
\label{complexity3}
The computational complexity of Algorithm~\ref{alg-ecoIso} is $O(mn^2/\epsilon^2)$.
\end{theorem}
\begin{proof}
The running time of Algorithm~\ref{alg-ecoIso} is dominated by Steps $2$ and $3$.
We first show that the total cost of Step $2$ is $O(mn^2/\epsilon^2)$.
In order to find the index $j_{i}$, by Step $2$ for all $\bfv\in\{\bfv_1, \bfv_2,
\ldots, \bfv_m\}$, one has to compute
$$
\begin{array}{ll}
\bfv^\top R_{i-1} \bfv&=\bfv^\top(I_n-L_{i-1})\bfv=\|\bfv\|_2^2-\bfv^\top L_{i-1}\bfv
\\&=\|\bfv\|_2^2-\bfv^\top (\alpha_1^{i-1}L_{i-2}+\alpha_2^{i-1} \bfv_{j_{i-1}}\bfv_{j_{i-1}}^\top)\bfv
\\
&=\|\bfv\|_2^2-\alpha_1^{i-1}\bfv^\top L_{i-2}\bfv-\alpha_2^{i-1}(\bfv_{j_{i-1}}^\top \bfv)^2.
\end{array}
$$
In the first iteration, we need $O(mn)$ flops to compute all $\|\bfv\|_2^2$ with $\bfv\in V$. For $i\geq 2$, since we have already computed $\bfv^\top L_{i-2}\bfv$ in the $(i-2)$-th iteration, so in the $i$-th iteration one only has to compute $\bfv_{j_{i-1}}^\top \bfv$ which need $O(mn)$ flops per-iteration.
After $\lceil\frac{n}{\epsilon^2}\rceil$ iterations, the total cost would be
$O(mn^2/\epsilon^2)$.

Next we will show that the
total cost of Step $3$ is $O(n^2/\epsilon^2)$. Indeed, each iteration needs to solve
a $2\times 2$ linear system $B_i\alpha^i=b_i$ with
$$
B_i=\left(
  \begin{array}{cc}
    \langle L_{i-1},L_{i-1}\rangle & \langle L_{i-1},\bfv_{j_i}\bfv_{j_i}^\top\rangle \\
    \langle\bfv_{j_i}\bfv_{j_i}^\top, L_{i-1}\rangle & \|\bfv_{j_i}\|^4_2 \\
  \end{array}
\right)\ \ \mbox{and}\ \ b_i=\left(\begin{array}{cc}
                  \langle I_n,L_{i-1}\rangle \\
                  \langle I_n,\bfv_{j_i}\bfv^\top_{j_i}\rangle
                \end{array}\right),
$$
 which can be efficiently solved if $B_i$ and $b_i$ are given. Note that $\langle L_{i-1},\bfv_{j_i}\bfv_{j_i}^\top\rangle$ and $\|\bfv_{j_i}\|_2$ had been already computed in Step $2$, and $\langle I_n,L_{i-1}\rangle$ can be computed in $O(n)$ flops. Next, we show that $ \langle L_{i-1},L_{i-1}\rangle$ can also be computed efficiently.  Indeed, according to Step $4$ of Algorithm \ref{alg-ecoIso}, we have
$$
\langle L_{i-1},L_{i-1}\rangle=(\alpha_1^{i-1})^2\langle L_{i-2},L_{i-2}\rangle+2\alpha_1^{i-1}\alpha_2^{i-1}\langle L_{i-2},\bfv_{j_{i-1}}\bfv_{j_{i-1}}^\top\rangle+(\alpha_2^{i-1})^2\|\bfv_{j_{i-1}}\|_2^4
$$
As we have already computed $\langle L_{i-2},L_{i-2}\rangle,\langle L_{i-2},\bfv_{j_{i-1}}\bfv_{j_{i-1}}^\top\rangle$ and $\|\bfv_{j_{i-1}}\|_2$, so we can compute $\langle L_{i-1},L_{i-1}\rangle$ in $O(1)$ time.
 Hence the cost in Step $3$ is $O(n^2/\epsilon^2)$ for total iteration number $\lceil\frac{n}{\epsilon^2}\rceil$.
\end{proof}

\begin{remark}

The computational time for Algorithm~\ref{alg-ecoIso} is  $O(mn^2/\epsilon^2)$ which is a good
improvement to the twice Ramanujan sparifiers in
\cite{BSS14} based on ``barrier'' potential function to guide the choice of indices as
its computational time is $O(mn^3/\epsilon^2)$. We remark that
the computational time can be reduced if massive
parallel processors are used. For example,  if a GPU with $m$ processes is used,
the computational time will be $O(n^2)$.
In addition, if $\bfv_i$ are sparse vectors, then the computational time can also be reduced.
For example, if $\|\bfv_i\|_0= O(\log n)$, the computational time will be $O(mn\log n)$.
\end{remark}

\subsection{Nonnegativity of $L_i$ \label{sec3.2}}
We next explain that $L_{i}$ obtained at each step in Algorithm~\ref{alg-ecoIso} has nonnegative coefficients. Indeed, we have the following theorem.
\begin{theorem}
\label{mjlai06122020}
Let $L_i$ be the symmetric matrix obtained from Algorithm~\ref{alg-ecoIso}.
Then
\begin{equation}
\label{Liformula}
L_i =\sum_{\ell=1}^i c_i(j_\ell) \bfv_{j_\ell}\bfv_{j_\ell}^\top
\end{equation}
 with $c_i(j_\ell)\ge 0, \ell=1,\cdots, i$ for any $i\geq 1$.
\end{theorem}

In the next subsection, we will show that $L_i$ is the desired approximation of the identity matrix $I_n$ satisfying \eqref{napp2} for $i\geq n/\epsilon^2$. In fact we will show that the residual matrix $R_i=I_n-L_i$ in Algorithm~\ref{alg-ecoIso}
satisfying $\|R_i\|_2\leq \epsilon/\beta$ for $i\geq n/\epsilon^2$.
To prove Theorem \ref{mjlai06122020}, let us begin with the following lemma.
\begin{lemma}
\label{orth} $\langle R_{i}, \bfv_{j_i}\bfv_{j_i}^\top\rangle=0$ and $\langle
R_i,L_{i-1}\rangle=0$.  Hence, $\langle R_i, L_i \rangle=0$.
\end{lemma}
\begin{proof}
The first two equations are simply the properties of the minimizer $L_i$. Since $L_i$ is a linear
combination of $\bfv_{j_i}$ and $L_{i-1}$, thus we have the last equation.
\end{proof}

\begin{lemma}
\label{independent}
Suppose that $R_{i}\neq 0$ for some $i\geq 1$. Then $\bfv_{j_{i+1}}\bfv_{j_{i+1}}^\top$ is linearly independent of $L_{i}$.
\end{lemma}
\begin{proof}
Suppose that $\bfv_{j_{i+1}}\bfv_{j_{i+1}}^\top$ is linearly dependent of $L_{i}$. Then there exists a nonzero coefficient $
\theta$, such that $\bfv_{j_{i+1}}\bfv_{j_{i+1}}^\top=\theta L_i$ which impiles
$$\bfv_{j_{i+1}}^\top R_i \bfv_{j_{i+1}} =  \langle R_i, \bfv_{j_{i+1}} \bfv_{j_{i+1}}^\top\rangle =\theta
\langle R_i, L_i\rangle =0,$$
 where the last equality follows from Lemma~\ref{orth}.
We claim that this indicates that $R_{i}=0$. Indeed, by \eqref{positive}, we have $\bfv_{j_{i+1}}^\top R_i \bfv_{j_{i+1}}=0$ which implies
$\bfv^\top R_i \bfv=0$ for all $\bfv\in \{\bfv_1,\cdots, \bfv_n\}$ and hence,  $\langle R_{i}, I_n\rangle=0$
since $\sum_{i\leq m}\bfv_i \bfv_i^\top =I_n$.
Together with $\langle R_i, L_{i}\rangle=0$, we have $\langle R_i, I_n-L_i\rangle=\langle R_i, R_i\rangle=0$, hence $R_{i}=0$. However, this contradicts the assumption that $R_{i}\neq 0$.  This completes the proof.
\end{proof}

Now we are  ready to establish the  nonnegativity of the coefficients of $L_i$, i.e., to prove Theorem \ref{mjlai06122020}.

\begin{proof}[Proof of Theorem \ref{mjlai06122020}] Indeed, it suffices for us to show that $\alpha^1_k,\alpha^2_k\geq 0$ for any $k=1,\ldots,i$. Our discussion is based on induction.
  For
$i=1$, we have $L_{1}= \alpha_2^1 \bfv_{j_1}\bfv_{j_1}^\top$ and $\alpha^1_1=0$.
To see $\alpha_2^1\geq0$, we expand the minimization in Step $3$ of Algorithm~\ref{alg-ecoIso} to have
\begin{eqnarray*}
\|I_n - \alpha \bfv_{j_1}\bfv_{j_1}^\top\|^2_F &=& n - 2\alpha \langle \bfv_{j_1}\bfv_{j_1}^\top, I_n\rangle+ \alpha^2
\langle \bfv_{j_1}\bfv_{j_1}^\top, \bfv_{j_1}\bfv_{j_1}^\top\rangle\cr
&=& n - 2\alpha \|\bfv_{j_1}\|_2^2 +\alpha^2 \|\bfv_{j_1}\|_2^4.
\end{eqnarray*}
Hence $\alpha_2^1$ satisfies $-2\|\bfv_{j_1}\|_2^2 + 2\alpha \|\bfv_{j_1}\|_2^4=0$ or $\alpha_2^1=1/\|\bfv_{j_1}\|_2^2 >0$.

We now assume that $L_{i-1} =\displaystyle \sum_{k=1}^{i-1} c_{i-1}(j_k) \bfv_{j_k}\bfv_{j_k}^\top$
with nonnegative coefficients $c_{i-1}(j_k)$ for $k=1, \cdots, i-1$.
Now let us take a look at the coefficients of $
L_{i}$ from Step $4$ of Algorithm~\ref{alg-ecoIso}.
The coefficients $\alpha_1^i, \alpha_2^i$ satisfy the following system of linear equations:
\begin{eqnarray}
\label{keyeq}
\alpha_1 \|L_{i-1}\|_F^2 + \alpha_2 \langle L_{i-1}, \bfv_{j_i}\bfv_{j_i}^\top \rangle
&=& \sum_{k=1}^{i-1}c_{i-1}(j_k)\|\bfv_{j_k}\|_2^2,\cr
\alpha_1 \langle L_{i-1}, \bfv_{j_i}\bfv_{j_i}^\top \rangle + \alpha_2  \|\bfv_{j_i}\|_2^4 &=& \|\bfv_{j_i}\|_2^2.
\end{eqnarray}
If $\bfv_{j_i}$ is linearly dependent of $L_{i-1}$,  by Lemma~\ref{independent},
we know $R_i=0$ which means that we have already found the desired subset.
Otherwise, it is clear that the coefficient matrix is nonsingular since the determinant is
$$
D=\|L_{i-1}\|_F^2 \|\bfv_{j_i}\|_2^4 - ( \langle L_{i-1}, \bfv_{j_i}\bfv_{j_i}^\top \rangle)^2>0,
$$
 where we have used Cauchy-Schwarz inequality as $\bfv_{j_i}\bfv_{j_i}^\top$ is linearly independent of $L_{i-1}$ by
Lemma~\ref{independent}. Using the Cramer's rule, we see that
\begin{eqnarray}
\label{keyformula}
\alpha^i_1 &=& D^{-1}\bigg( \|\bfv_{j_i}\|_2^4 \sum_{k=1}^{i-1}c_{i-1}(j_k)\|\bfv_{j_k}\|_2^2 -
\langle L_{i-1}, \bfv_{j_i}\bfv_{j_i}^\top \rangle \|\bfv_{j_i}\|_2^2 \bigg),\cr
\alpha^i_2 &=& D^{-1}\bigg( -\langle L_{i-1}, \bfv_{j_i}\bfv_{j_i}^\top \rangle
\sum_{k=1}^{i-1}c_{i-1}(j_k)\|\bfv_{j_k}\|_2^2+ \|L_{i-1}\|_F^2  \|\bfv_{j_i}\|_2^2 \bigg).
\end{eqnarray}
It is easy to see that
$$\alpha_1^i= D^{-1}\|\bfv_{j_i}\|_2^2 \sum_{k=1}^{i-1}c_{i-1}(j_k)
\big(\|\bfv_{j_i}\|_2^2 \|\bfv_{j_k}\|_2^2 -\langle \bfv_{j_k}\bfv_{j_k}, \bfv_{j_i}\bfv_{j_i}^\top\rangle\big)\ge 0.$$
It remains to  show that $\alpha_2^i$ is nonnegative.
Let us take a close look at the right-hand side of $\alpha_2^i$ in \eqref{keyformula}. By Lemma~\ref{orth}, we always have
\begin{equation}
\label{ugakey}
\langle I_n- L_{i-1}, L_{i-1}\rangle=0 \hbox{ or } \|L_{i-1}\|_F^2= \langle L_{i-1}, I_n\rangle.
\end{equation}
Thus, we have
\begin{eqnarray*}
D \alpha_2^i&=& -\langle L_{i-1}, \bfv_{j_i}\bfv_{j_i}^\top \rangle
\sum_{k=1}^{i-1}c_{i-1}(j_k)\|\bfv_{j_k}\|_2^2+ \|L_{i-1}\|_F^2  \|\bfv_{j_i}\|_2^2\cr
&=& -\langle L_{i-1}, \bfv_{j_i}\bfv_{j_i}^\top \rangle\|L_{i-1}\|_F^2 + \|L_{i-1}\|_F^2
\langle I_n, \bfv_{j_i}\bfv_{j_i}^\top\rangle \cr
&=& \bfv_{j_i}^\top R_{i-1} \bfv_{j_i} \|L_{i-1}\|_F^2 ,
\end{eqnarray*}
where we have used the fact (\ref{ugakey}). So $\alpha_2^i$ will be nonnegative if $\bfv_{j_i}^\top R_{i-1} \bfv_{j_i}\ge 0$.

We claim that this is true for all $i\ge 1$.
For $i=1$,  we have $\bfx^\top R_0 \bfx\ge 0$ since $R_0=I_n$.
For $i=2$, it is easy to see that $c_1(j_1)= 1/\|\bfv_{j_1}\|_2^2$. We have $$\bfv_{j_2}^\top R_{1} \bfv_{j_2}
= \bfv_{j_2}^\top \bfv_{j_2} -  \frac{(\bfv_{j_1}^\top \bfv_{j_2})^2}{\|\bfv_{j_1}\|_2^2}\ge 0.$$
We show that $\bfv_{j_{i+1}}^\top R_{i} \bfv_{j_{i+1}}\ge 0$ for all $i\ge 2$. In fact, it is easy to see that
$$
 \bfv_{j_i}^\top R_{i} \bfv_{j_i}= \|\bfv_{j_i}\|_2^2
- \alpha_1^i \bfv_{j_i}^\top L_{i-1}\bfv_{j_i}
-\alpha_2^i \|\bfv_{j_i}\|_2^4 = 0
$$
by using (\ref{keyeq}).
By the definition of $\bfv_{j_{i+1}}$, we have
 $$\bfv_{j_{i+1}}^\top R_{i} \bfv_{j_{i+1}} = \max_{\bfv\in V}
\bfv^\top R_i \bfv \ge \bfv_{j_i}^\top R_{i} \bfv_{j_i}= 0.$$
By (\ref{ugakey}), we have $\alpha_2^i\geq 0$.
We have therefore proved Theorem \ref{mjlai06122020}.
\end{proof}

\subsection{Convergence of Algorithm~\ref{alg-ecoIso}}
To establish the convergence of Algorithm~\ref{alg-ecoIso}, let us first  introduce
some notations and concepts. Set
$\gamma:=\frac{\min_k\|\bfv_k\|_2^2}{\max_{k}\|\bfv_k\|_2^2}$. Since
$$
m\beta \max_k\|\bfv_k\|_2^2 \leq \sum_{k=1}^m \|\bfv_k\|_2^2=n,
$$
we have
\begin{equation}
\label{anotherestimate}
\max_k\|\bfv_k\|_2^2\leq \frac{1}{\gamma}\frac{n}{m}.
\end{equation}
Next define
$$\mathcal{F}=\bigg\{ S: S\subset \{1, \cdots, m\}, \sum_{i\in S} s_i\bfv_i \bfv_i^\top= I_n, s_i>0, i\in S\big\}$$
to be the feasible set of all possible isotropic subsets in $V$.  It is clear that
$\mathcal{F}$ is not  empty since the whole set $\{1, \cdots, m\}\in \mathcal{F}$.  We  set
\begin{equation}
\label{KSsize}
K:=\min_{S\in \mathcal{F}}K_S
\end{equation}
 where $K_S= \sum_{i\in S} s_i. $ Since $\{1, \cdots, m\}\in \mathcal{F}$, $K\le m$.

To establish the convergence of Algorithm~\ref{alg-ecoIso}, we begin with the following two lemmas.
\begin{lemma}
\label{eco2}
$\|R_i\|_F^2\le \|R_{i-1}\|_F^2- \frac{(\bfv_{j_i}^\top R_{i-1}\bfv_{j_i})^2}{\|\bfv_{j_i}\|_2^4}$ for all $i\geq 1$.
\end{lemma}
\begin{proof}
For all $i\geq 1$, we have
$$
\begin{aligned}
\|R_i\|_F^2&=\min\limits_{\alpha_1, \alpha_2}
\|R_{i-1}-(\alpha_1-1)L_{i-1}-\alpha_2\bfv_{j_i}\bfv_{j_i}^\top \|^2_F\\
&\leq \min_{\alpha_2}\|R_{i-1}-\alpha_2\bfv_{j_i}\bfv_{j_i}^\top\|_F^2=\|R_{i-1}\|_F^2-
\frac{(\bfv_{j_i}^\top R_{i-1}\bfv_{j_i})^2}{\|\bfv_{j_i}\|_2^4}.
\end{aligned}
$$
\end{proof}

Next we need a classic  elementary lemma. For the completeness, we include an elementary induction proof here.
\begin{lemma}[DeVore and Temlyakov, 1996 \cite{1996Some}]
\label{DT96}
Suppose we have two sequences of  nonnegative numbers $\{a_k, k\ge 1\}$ and $\{\beta_k, k\ge 1\}$
satisfying $\beta_k>0$ and $a_1= 1$ and
$$a_{k+1}\le a_k (1- a_k \beta_k), \quad \text{ for any } k\ge 1.$$
Then
$$
a_{i+1}\le \frac{1}{1+\sum_{k=1}^i \beta_k}, \quad \text{ for any } i\ge 0.
$$
\end{lemma}
\begin{proof}
We prove the conclusion by induction. Suppose that we have
\[
a_{i}\le \frac{1}{1+\sum_{k=1}^{i-1} \beta_k}
\]
for some $i\geq 1$. Then
$$
a_{i+1}^{-1} \geq a_i^{-1} (1- a_i \beta_i)^{-1}\geq a_i^{-1}(1+a_i\beta_i)=a_i^{-1}+\beta_i
\geq 1+\sum_{k=1}^{i-1}\beta_k+\beta_i,
$$
which implies the conclusion.
\end{proof}

We are finally ready to establish the main result in this section.
\begin{theorem}
\label{mainresult41} Let $K$ be the size of subset given in \eqref{KSsize}. Then each
step in Algorithm~\ref{alg-ecoIso} satisfies
\begin{equation}
\label{theorem36-1}
\|R_{i+1}\|^2_F\leq \frac{K^2}{K^2/n +\sum_{k=1}^i 1/\|\bfv_{j_k}\|_2^4}, \ \ \forall i\geq 1.
\end{equation}
\end{theorem}
\begin{proof}
By Lemma~\ref{orth}, we have
 $$
 \|R_i\|_F^2 =\langle R_i, R_i\rangle= \langle R_i, I_n\rangle
 =  \sum_{k\in S} s_k  \bfv_k^\top R_i \bfv_k.
 $$
 where $S\in \mathcal{F}$ satisfying $ K=K_S$.

Since  $\bfv_{j_{i+1}}^\top R_i \bfv_{j_{i+1}}\ge 0$,  by Step $2$ of
Algorithm~\ref{alg-ecoIso}, we have
 $$
 \|R_i\|_F^2 \le K \bfv_{j_{i+1}}^\top R_i\bfv_{j_{i+1}} .
 $$
 It follows from Lemma~\ref{eco2} that
$$\begin{array}{ll}
\|R_i\|_F^2 &\leq
\|R_{i-1}\|_F^2- \frac{(\bfv_{j_i}^\top R_{i-1}\bfv_{j_i})^2}{\|\bfv_{j_i}\|_2^4}
\\
&\leq
\|R_{i-1}\|_F^2 -\frac{ \|R_{i-1}\|_F^4}{\|\bfv_{j_i}\|_2^4K^2}
\\
&=
\|R_{i-1}\|_F^2\big(1 - \frac{1}{\|\bfv_{j_i}\|_2^4 K^2} \|R_{i-1}\|_F^2\big).
\end{array}
$$
Taking $a_k=\frac{\|R_k\|_F^2}{n}$ and $\beta_k= \frac{n}{\|\bfv_{j_k}\|_2^4 K^2}$ in
Lemma \ref{DT96}, we arrive at the conclusion.
\end{proof}

\begin{corollary}
\label{case1}
Suppose that $K$ satisfies $K\le m/\sqrt{n}$.  Then when $i\geq \frac{n}{\epsilon^2}$, we have
$$
\|R_{i+1}\|_F\leq \epsilon/\gamma,
$$
where $\gamma$ is defined in (\ref{anotherestimate}).
\end{corollary}
\begin{proof}
Recall from (\ref{anotherestimate}) we have
$$
\max_k \|\bfv_k\|_2^4 \le \left( \frac{n}{\gamma m}\right)^2.
$$
The right-hand side of (\ref{theorem36-1}) can be estimated as follows:
\begin{eqnarray*}
\frac{K^2}{K^2/n +\sum_{k=1}^i 1/\|\bfv_{j_k}\|_2^4} &\le&
\frac{K^2}{K^2/n+ i/\max_k \|\bfv_k\|_2^4}
\le K^2 \left( \frac{n}{\beta m}\right)^2/i \cr
&\le& K^2 \left( \frac{1}{\beta m}\right)^2 n \epsilon^2 \le \epsilon^2/\gamma^2.
\end{eqnarray*}
It follows that $\|R_{i+1}\|_F \le \epsilon/\gamma$.
\end{proof}

Once  $\|R_{i+1}\|_F\leq \epsilon/\gamma$, we have $\|R_{i+1}\|_2\leq \epsilon/\gamma$
immediately which leads to  the sparse
subset satisfying a desired estimate similar to (\ref{app2}).
Indeed, recalling $R_{i+1}= I_n- L_{i+1}$ with $L_{i+1}$ being
the sparsifier which is given in (\ref{Liformula}), we have
$$
\bfx^\top L_{i+1}\bfx = \bfx^\top \bfx + \bfx^\top (L_{i+1}- I_n)\bfx \le (1+ \epsilon/\gamma)\bfx^\top \bfx.
$$
Similarly, we can have the other hand side of the estimate
$\bfx^\top L_{i+1}\bfx\geq (1- \epsilon/\gamma)\bfx^\top \bfx$.
These estimates are desired estimates when $\gamma=1$. Even $\gamma\not
=1$, these estimates show that $L_i$ is a good sparsifier satisfying the following
\begin{equation}
\label{app6}
(1-\epsilon/\gamma)I_n \preceq  L_{i+1}\preceq (1+ \epsilon/\gamma)I_n,
\end{equation}
where $\epsilon>0$ is small enough such that $1-\epsilon/\gamma>0$.

In general, we do not know when $K\le m/\sqrt{n}$ happens. Instead,  we
have only $K\le m$ as explained above. Assuming $K=m$, using a similar proof of Corollary~\ref{case1} we have
$$
\|R_{i+1}\|_F^2\le n \epsilon^2/\gamma^2 \hbox{ or } \|R_{i+1}\|_F \le \sqrt{n} \epsilon/\gamma
$$
after $i\ge n/\epsilon^2$ iterations.
Since $\|R_{i+1}\|_F \le \sqrt{n} \|R_{i+1}\|_2$, we heuristically expect
$\|R_{i+1}\|_2 \approx  \epsilon/\gamma$
although the this estimate has not been proved when the case $K=m$. Let us present an example $K<m$ to show that
the assumption in Corollary~\ref{case1} can happen.
\begin{example}
Let $\bfe_i\in \mathbb{R}^n$ be the standard unit vector with $1$ on the $i$th entry and zero otherwise, where $i=1, 2,
\cdots, n$. We first choose $ \bfe_i/\sqrt{2}, i=1, \cdots, n$ and then choose $(n-1)$ copies of $\{\bfe_i/\sqrt{2(n-1)}$
to form $V$. Then
$$
I_n=  \sum_{i=1}^n \frac{1}{2}\bfe_i \bfe_i^\top + \sum_{k=1}^{n-1} \sum_{i=1}^n \frac{1}{2(n-1)}\bfe_i \bfe_i^\top.
$$
We can easily see that $m=n^2$ and $K=2n$ which is less than $m/\sqrt{n}$ for $n\ge 4$.
\end{example}

\subsection{Subset selection for general matrix $B$}
We now return to the original problem for finding a sparse positive subset for
general symmetric positive semidefinite matrix $B =\displaystyle \sum_{i=1}^m \bfv_i
\bfv_i^\top$. Since $B$ is symmetric, we can write $B=L D L^\top$ with $D$ being the diagonal matrix.
Indeed, let us first consider the case that $B$ is invertible. Then $D=I_n$ and $L$ can be a low-triangular matrix (Cholesky decomposition), thus
\begin{equation}
\label{newIn}
I_n = \sum_{k=1}^m L^{-1}\bfv_k (L^{-1}\bfv_k)^\top.
\end{equation}
 We remark that the factorization of $B=L L^\top$ needs a computational time
$O(n^3)$ and form the new vectors $L^{-1}\bfv_k, k=1, \cdots, m$ is $O(mn^2)$ which is similar to
the computational complexity  $O(mn^2/\epsilon^2)$ of
Algorithm~\ref{alg-ecoIso} as $m\geq n$. That is,  the order of computational complexity
for a general symmetric positive definite matrix $B$ does not increase.


Next let us consider the case that $B$ is not full rank and assume $\mbox{rank}(B)=r$. We can still find a factorization $B=L D
L^\top$ with $D$ being a diagonal matrix.  Indeed, if $B=A^\top A$ for $A$ of size $m\times n$,
we can use the thin SVD in $O(mr^2)$ time to find $USV^\top$ of rectangular matrix $A$ so that $B= LD L^\top $ with $L=V\in \mathbb{R}^{n\times r}$ and
$D=S^\top S\in \mathbb{R}^{r\times r}$. Note that $L^\top L=V^\top V=I_r$, then we can rewrite $B=\sum\limits_{i=1}^m\bfv_i\bfv_i^\top$ as
\begin{equation}
\label{newIn2}
I_r =  \sum_{k=1}^m S^{-1} V^\top \bfv_k(S^{-1}V^\top \bfv_k)^\top,
\end{equation}
where $I_r$ is the identity matrix of size $r \times r$. The computational cost of $S^{-1} V^\top \bfv_k,k=1,\ldots,m$ is
$O(mnr)$. Applying Algorithm~\ref{alg-ecoIso} to those new vectors $S^{-1} V^\top \bfv_k, k=1, \cdots, m$, the computational
complexity is $O(mr^2/\epsilon^2)$. That is, the total computational cost will not increase as $r\leq n$.

Now we can define $\gamma$ and $K$ similarly as in the subsection above.
Letting ${\bf u}_k=L^{-1}\bfv_k$ or ${\bf u}_k=S^{-1}V^\top\bfv_k$, define
\begin{equation}
\label{newformulae}
\gamma= \frac{\min_k \|{\bf u}_k\|_2}{\max_k \|{\bf u}_k\|_2} \hbox{ and }
K= \min_{S\in\mathcal{F}} \sum_{i\in S} s_i,
\end{equation}
where $\mathcal{F}=\{S:S\subset\{1,\ldots,m\},\sum\limits_{i\in S}s_i{\bf u}_i{\bf u}_i^\top=I,s_i>0,i\in S\}$.
Armed with these $\beta$ and $K$, we are able to prove the following
\begin{theorem}
\label{mainresult4} Suppose that $\bfv_1, \cdots, \bfv_m\in \mathbb{R}^n$ with $B
=\displaystyle \sum_{i=1}^m \bfv_i \bfv_i^\top$. Suppose that $K\le m/\sqrt{\mbox{rank}(B)}$.
Then for any $\epsilon>0$, Algorithm~\ref{alg-ecoIso} can find $s_i\ge 0$ with $|\{i: s_i \not=0\}|\le \hbox{rank}(B)/\epsilon^2$ such that
\begin{equation}
(1- \epsilon/\gamma) B \preceq \sum_{i=1}^m s_i \bfv_i \bfv_i^\top  \preceq (1+\epsilon/\gamma)  B.
\label{app4}
\end{equation}
\end{theorem}
\begin{proof}
 As explained above, we are able to have
\eqref{newIn} or \eqref{newIn2}. Then
we apply Algorithm~\ref{alg-ecoIso} to have positive coefficients $c_i(j_\ell)$ and sub-indices $j_\ell$ such that
$$\|I_n-  \sum_{\ell=1}^i c_i(j_\ell) {\bf u}_{j_\ell} {\bf u}_{j_\ell} ^\top\|_F\le \epsilon/\gamma$$
by Corollary~\ref{case1} and hence, for all $\bfx\in \mathbb{R}^{\mbox{rank}(B)}$,
$$
(1- \epsilon/\gamma) \bfx^\top \bfx \le \bfx^\top  \sum_{\ell=1}^i c_i(j_\ell) {\bf u}_{j_\ell}
{\bf u}_{j_\ell} ^\top \bfx \le (1+\epsilon/\gamma)  \bfx^\top \bfx
$$
as explained in the previous subsection.  If $B$ is invertible, letting $\bfx=L^\top \bfy$, we have
$$
(1- \epsilon/\gamma) \bfy^\top B\bfy \le \bfy^\top L \sum_{\ell=1}^i c_i(j_\ell) (L^{-1}\bfv_{j_\ell}) (L^{-1}\bfv_{j_\ell})^\top
L^\top\bfy \le (1+\epsilon/\gamma)  \bfy^\top B\bfy,
$$
which is (\ref{app4}). If $B$ is not full rank, letting $\bfx=(VS)^\top \bfy$, we have
\begin{equation}
\label{2021-1-28-1}
(1- \epsilon/\gamma) \bfy^\top B\bfy \le \bfy^\top \sum_{\ell=1}^i c_i(j_\ell) (VV^\top\bfv_{j_\ell}) (VV^\top\bfv_{j_\ell})^\top
\bfy \le (1+\epsilon/\gamma)  \bfy^\top B\bfy.
\end{equation}
Note that $VV^\top$ is the projection onto the range of $[\bfv_1,\ldots,\bfv_m]$, hence $VV^\top\bfv_{j_\ell}=\bfv_{j_\ell}$
which implies that \eqref{2021-1-28-1} is indeed \eqref{app4}. These complete the proof.
\end{proof}


\section{The UGA algorithm for sparsifiers}
\label{sect:eco-omp}
In this section, we consider the graph sparsification problem.
Following the ideas in \cite{BSS14}, we can use the result of
Theorem~\ref{mainresult4} to compute the sparsifier $H_i$ of a given graph $G$
in $O(mn^2/\epsilon^2)$ time. Instead,
we propose another algorithm to compute a graph sparsifier which can be done
in $O(m+n^2/\epsilon^2)$ which is much faster.
We shall first present our computational algorithm in this section.
Then we  explain its computational complexity.
Next we present the convergence analysis of the algorithm.
Finally, we shall show that the subgraph obtained from the UGA algorithm is indeed a graph sparsifier.

\subsection{An UGA for graph sparsification}
Let $G=(V,E,{\bf w})$ be a weighted undirected graph.
To state conveniently, we define
$$\phi_{(u,v)}:=(\bfe_u-\bfe_v)(\bfe_u-\bfe_v)^\top,$$
where $\bfe_u\in \mathbb{R}^n$ is a standard basis vector which is zero everywhere except for the
$u$-th component which is $1$.
Then the Laplacian $L_G$ can be simple described as
$$L_G=\sum\limits_{(u,v)\in E}w_{(u,v)} \phi_{(u,v)}.$$
Note that if $w_{(u,v)}=1$, then $L_G$ is the standard graph Laplacian for an undirected graph $G=(V,E)$, and if $w_{(u,v)}>0$, $L_G$ is a weighted graph Laplacian.

The UGA for spectral sparsification is given in Algorithm \ref{alg-eco}.

\begin{algorithm}[htpb]
\caption{Universal greedy  approach (UGA) for graph sparsification \label{alg-eco}}
\begin{algorithmic}
\Require
Laplacian matrix $L_G\in\mathbb{R}^{n\times n}$, $\epsilon\in(0,1)$.
\begin{enumerate}
\item[1:] $R_0:=L_G$, $\Lambda_0:=\emptyset$, $L_{H_0}:=0$, $C_0:=0\in\mathbb{R}^{|E|}$ and $i:=1$.
\item[2:] Find an edge $(u_i,v_i)$ such that
$$
\big|\langle \phi_{(u_i,v_i)},R_{i-1}\rangle\big|=\max_{(u,v)\in E}\big| \langle \phi_{(u,v)},R_{i-1}\rangle\big|,$$
and update $\Lambda_i=\Lambda_{i-1}\cup\{(u_i,v_i)\}$.
\item[3:] Compute the optimal weights
\begin{equation}
\label{Aeco1}
(\alpha_1^i,\alpha_2^i)=\arg\min\limits_{(\alpha_1,\alpha_2)\in\mathbb{R}^2}\big\|L_G-\alpha_1L_{H_{i-1}}-\alpha_2\phi_{(u_i,v_i)}\big\|^2_F.
\end{equation}
\item[4:] Update
\begin{equation}
\label{LH}
L_{H_{i}}=\alpha_1^iL_{H_{i-1}}+\alpha^i_2\phi_{(u_i,v_i)}
\end{equation}
and
\begin{equation}
\label{residual:eco}
R_i=L_G-L_{H_{i}}.
\end{equation}
\item[5:] Update the coefficient $C_i=\alpha_1^iC_{i-1}$ and set $C_i(u_i,v_i)=\alpha_2^i+\alpha_1^iC_{i-1}(u_i,v_i)$.
\item[6:] If $i>\lceil\frac{n}{\epsilon^2}\rceil$, stop and go to output. Otherwise, set $i=i+1$ and return to Step $2$.
\end{enumerate}

\Ensure
The sparsifier $L_{\lceil\frac{n}{\epsilon^2}\rceil}$, the selected edges $\tilde{E}=\Lambda_i$  and the weight function $\tilde{\bf w}=C_i$.
\end{algorithmic}
\end{algorithm}

\subsection{Computational complexity}
The running time of Algorithm~\ref{alg-eco} is dominated by Steps $2$ and $3$,
whose total cost is $O(m +n^2/\epsilon^2)$. In terms of the number $m$ of edges, the total cost is $O(m)$
if $m = O(n^2)$.
\begin{theorem}
\label{complexity2}
The computational complexity of Algorithm~\ref{alg-eco} for graph $G$ with $n$ vertices and $m$ edges
is $O(m+n^2/\epsilon^2)$.
\end{theorem}
\begin{proof}
 We first show that the total cost of Step $2$ is
$O(m+n^2/\epsilon^2)$. By Step $2$, we have
$$
\big\langle \phi_{(u,v)},R_{i}\big\rangle =\big\langle \phi_{(u,v)},L_G-L_{H_{i}}\big\rangle
= \big\langle \phi_{(u,v)},L_G\big\rangle-\big\langle \phi_{(u,v)},L_{H_{i}}\big\rangle.
$$
Since each $\phi_{(u,v)}$ has only four nonzeor entries, so in the first iteration, we need $O(m)$ flops to compute all $\big\langle \phi_{(u,v)},L_G\big\rangle$ with $(u,v)\in E$. Note that
$$
\big\langle \phi_{(u,v)},L_{H_{i}}\big\rangle=\alpha_1^i\big\langle \phi_{(u,v)},L_{H_{i-1}}\big\rangle+\alpha^i_2\big\langle
\phi_{(u,v)}, \phi_{(u_i,v_i)}\big\rangle,
$$
hence one only need to compute the last term
$\big\langle \phi_{(u,v)}, \phi_{(u_i,v_i)}\big\rangle$ with using an incremental method.
Noting that
$$
\langle \phi_{(u,v)}, \phi_{(u_i, v_i)}\big\rangle=
\left\{
\begin{array}{ll}
0, \ \ \ \ \ \mbox{if} \ \ u\neq u_i \ \ \mbox{and} \ \ v\neq v_i,
\\
4, \ \ \ \ \ \mbox{if} \ \ u = u_i \ \ \mbox{and} \ \ v= v_i,
\\
1, \ \ \ \ \ \mbox{if} \ \ u = u_i \ \ \mbox{and} \ \ v\neq v_i,
\\
1, \ \ \ \ \ \mbox{if} \ \ u\neq u_i \ \ \mbox{and} \ \ v = v_i,
\end{array}
\right.
$$
hence the $i$-th subsequent iteration need $O(n)$ flops to calculate all $\langle
\phi_{(u,v)}, \phi_{(u_i, v_i)}\big\rangle$.
Thus the  total computational  cost of Step $2$ is $O(m+n^2/\epsilon^2)$ as the total number of iterations is
$\lceil\frac{n}{\epsilon^2}\rceil$.

Now let us discuss the computational cost of Step $3$. Set
$$
B_{i}=\left(
        \begin{array}{cc}
        \big\langle L_{H_{i-1}},L_{H_{i-1}}\big\rangle & \big\langle L_{H_{i-1}}, \phi_{(u_i,v_i)}\big\rangle \\
          \big\langle L_{H_{i-1}},\phi_{(u_i,v_i)}\big\rangle & \big\langle \phi_{(u_i,v_i)}, \phi_{(u_i,v_i)}\big\rangle \\
        \end{array}
      \right)
$$
and $b_{i}=(\big\langle L_{G},L_{H_{i-1}}\big\rangle,\big\langle L_{G},
\phi_{(u_i,v_i)}\big\rangle)^\top$. Then in Step $3$, $\alpha^i$ can be compute by solving the following $2\times 2$ system of
linear equations
$$
B_{i}\alpha^i=b_{i}.
$$
Similar to the arguments of Theorem \ref{complexity3},  it can be efficiently solved in $O(n)$ time.
After running $\lceil\frac{n}{\epsilon^2}\rceil$ iterations, the time complexity for Step $3$ is
$O(n^2/\epsilon^2)$.
\end{proof}

\subsection{Convergence analysis}

In this subsection, we show that Algorithm~\ref{alg-eco} is convergent and the output matrix is a well  approximation of $L_G$.
For convenience,
let $S\subseteq E$ and $C\in \mathbb{R}_{+}^{|S|}$, define
$$
L_{S,C}:=\sum\limits_{(u,v)\in S} C(u,v) \phi_{(u,v)}.
$$
For any fixed $\epsilon$ with $\epsilon\in(0,1)$ and $\lceil\frac{n}{\epsilon^2}\rceil\le |E|$,
let $L_{n,\epsilon}$ be set of the best approximation of $L_G$ having at most $\lceil\frac{n}{\epsilon^2}\rceil$ edges.
That is,
\begin{equation}\label{Lne}
L_{n,\epsilon}:=\bigg\{\sum\limits_{(u,v)\in S_{\epsilon}} C_{\epsilon}(u,v) \phi_{(u,v)}:(S_{\epsilon},C_{\epsilon})\in \Theta \bigg\}
\end{equation}
with
$$
\begin{array}{ll}
\Theta:=\hbox{arg}\min\limits_{S,C}\big\{ &\big\|L_G- L_{S,C} \big\|_F: \ (1-\epsilon)^2L_G\preceq L_{S,C}\preceq(1+\epsilon)^2L_G, \\
& |S|\leq  \lceil\frac{n}{\epsilon^2}\rceil \hbox{ and } C\geq 0
\big\}.
\end{array}
$$
Based on the result on graph sparsification in \cite{BSS14}, there is an $L_\epsilon\in L_{n,\epsilon}$, it
achieves the $\frac{(1+\epsilon)^2}{(1-\epsilon)^2}$-approximation of $G$.
The following result states that $L_{H_{i}}$ can approximate $L_{\epsilon}$ well.
\begin{theorem}
\label{mainresult}
Let $R_i$ be the residual matrix defined in \eqref{residual:eco} of Algorithm~\ref{alg-eco}. Then for any $L_\epsilon\in L_{n,\epsilon}$, either
$\|R_{i}\|_F\le \|L_G- L_\epsilon\|_F$ or
\begin{equation}
\label{eq:mainresult}
\|L_{H_i}-L_{\epsilon}\|_F \le 2\|L_G- L_\epsilon\|_F + \delta^i \|L_G\|_F, \quad \forall \ i\ge 0,
\end{equation}
where $\delta=\min\bigg\{\sqrt{1- \frac{1}{2|E|}},\sqrt{1-\frac{\epsilon^2}{4n}}\bigg\}$ is a constant belongs to $(0, 1)$.
\end{theorem}

%
%
Before we pass to the proof of this theorem, let us state several useful properties of Algorithm \ref{alg-eco}.
\begin{lemma}
\label{lemma-eco1}
$\langle R_i,\phi_{(u_i,v_i)}\rangle=0$ and $\langle R_i,L_{H_{i-1}}\rangle=0$.
\end{lemma}
\begin{proof}
Recall that $\alpha^i$ is the optimal solution of problem \eqref{Aeco1}.
By the first-order optimality condition according to $\phi_{(u_i,v_i)}$ and $L_{H_{i-1}}$, we have
$$
\langle L_G-\alpha_1^iL_{H_{i-1}}-\alpha^i_2\phi_{(u_i,v_i)},L_{H_{i-1}}\rangle=0
$$
and
$$
\langle L_G-\alpha_1^iL_{H_{i-1}}-\alpha^i_2\phi_{(u_i,v_i)},\phi_{(u_i,v_i)}\rangle=0
$$
which together with $R_{i}=L_G-\alpha_1^iL_{H_{i-1}}-\alpha^i_2\phi_{(u_i,v_i)}$ implies that $\langle R_i, \phi_{(u_i,v_i)}\rangle=0$ and $\langle R_i,L_{H_{i-1}}\rangle=0$.
\end{proof}

\begin{lemma}
\label{lemma-eco2}
$\|R_i\|_F^2=\|L_G\|_F^2-\|L_{H_{i}}\|_F^2$ for all $i\geq 0$.
\end{lemma}
\begin{proof}
For $i\geq0$,
$$
\begin{array}{ll}
\|L_G\|_F^2&=\|R_i+L_{H_{i}}\|^2_F =\|R_i\|^2_F+\|L_{H_{i}}\|^2_F+2\langle R_i,L_{H_{i}}\rangle
=\|R_i\|^2_F+\|L_{H_{i}}\|^2_F,
\end{array}
$$
the last equality follows from Lemma \ref{lemma-eco1} that $$\langle R_i,L_{H_{i}}\rangle=\alpha_1^i\langle R_i,L_{H_{i-1}}\rangle+\alpha_2^i\langle R_i,\phi_{(u_i,v_i)}\rangle=0.$$ This completes the proof of this lemma.
\end{proof}

\begin{lemma}
\label{lemma-eco4}
If $L_{H_{i-1}}=\beta \phi_{(u_i,v_i)}$ with nonzero $\beta$, then $\|R_{i}\|_F=\|R_{i-1}\|_F$.
\end{lemma}
\begin{proof}
If $L_{H_{i-1}}=\beta \phi_{(u_i,v_i)}$ for some $\beta\neq 0$, we get
\begin{equation}
\label{lemma-eco4-p1}
\begin{array}{ll}
\|R_{i}\|_F^2&=\min\limits_{\alpha\in\mathbb{R}^2}\big\|L_G-\alpha_1L_{H_{i-1}}-
\alpha_2\phi_{(u_i,v_i)}\big\|^2_F
\\
&=\min\limits_{\alpha\in\mathbb{R}^2}\big\|L_G-(\alpha_1+\alpha_2/ \beta)L_{H_{i-1}}\big\|^2_F
\\
&=\min\limits_{\gamma\in\mathbb{R}}\big\|L_G-\gamma L_{H_{i-1}}\big\|^2_F
\\
&=\min\limits_{\gamma\in\mathbb{R}}\big\|L_G-\gamma\alpha_1^{i-1}L_{H_{i-2}}-
\gamma\alpha_2^{i-1}\phi_{(u_{i-1},v_{i-1})}\big\|^2_F
\\
&\leq \min\limits_{ (\gamma_1,\gamma_2)\in\mathbb{R}^2}\big\|L_G-\gamma_1L_{H_{i-2}}
-\gamma_2\phi_{(u_{i-1},v_{i-1})}\big\|^2_F
\\
&=\|L_G-L_{H_{i-1}}\|^2_F=\|R_{i-1}\|_F^2
\end{array}
\end{equation}
and hence the conclusion  $\|R_{i}\|_F\leq\|R_{i-1}\|_F$ holds in this case. In general,
$$
\begin{array}{ll}
\|R_{i-1}\|_F^2&=\big\|L_G-\alpha^{i-1}_1L_{H_{i-2}}-\alpha_2^{i-1}
\phi_{(u_{i-1},v_{i-1})}\big\|^2_F
\\
&\geq\min\limits_{\alpha\in\mathbb{R}^2}
\big\|L_G-\alpha_1(\alpha^{i-1}_1L_{H_{i-2}}+\alpha_2^{i-1}
\phi_{(u_{i-1},v_{i-1})})-\alpha_2
\phi_{(u_i,v_i)}\big\|^2_F
\\
&=\min\limits_{\alpha\in\mathbb{R}^2} \big\|L_G-\alpha_1L_{H_{i-1}}-\alpha_2
\phi_{(u_i,v_i)}\big\|^2_F =\|R_{i}\|_F^2.
\end{array}
$$
This completes the proof.
\end{proof}

\begin{lemma}
\label{lemma3.6}
Suppose that $R_{i-1}\neq 0$ for some $i\geq1$. Then, $L_{H_{i-1}}\neq\beta \phi_{(u_i,v_i)}$ for all $\beta \neq 0$.
\end{lemma}
\begin{proof}
If $L_{H_{i-1}}=\beta \phi_{(u_i,v_i)}$ with $\beta \neq 0$,
similar to \eqref{lemma-eco4-p1} we have
$$
\begin{array}{ll}
\|R_{i}\|_F^2&=\min\limits_{\gamma\in\mathbb{R}}\big\|L_G-\gamma L_{H_{i-1}}\big\|^2_F\\
&=\big\|L_G-\gamma_{i-1}^* L_{H_{i-1}}\big\|^2_F =\|R_{i-1}\|_F^2 =\|L_G-L_{H_{i-1}}\|^2_F,
\end{array}
$$
where $\gamma_{i-1}^*$ denotes the optimal solution of the minimization in terms of $\gamma$ and
the third equality follows from Lemma \ref{lemma-eco4}. As $R_{i-1}\neq 0$ , we have
$L_{H_{i-1}}\neq L_G$. Then from the above equality, we conclude that $\gamma_{i-1}^*=1$ is the
unique optimal solution. While by its first-order optimality condition, we have
$$
\big\langle L_{H_{i-1}}-L_G,L_{H_{i-1}}\big\rangle=0, \hbox{ i.e.}, \
\big\langle R_{i-1},L_{H_{i-1}}\big\rangle=0.
$$
However, this contradicts
$$
\big|\big\langle R_{i-1},L_{H_{i-1}}\big\rangle\big|=\big|\beta\big\langle R_{i-1},
\phi_{(u_i,v_i)}\big\rangle\big|=|\beta|\max_{(u,v)\in E}\big| \langle \phi_{(u,v)},R_{i-1}\rangle\big|\neq 0.
$$
This completes the proof.
\end{proof}

Similar to Lemma \ref{eco2},
we can build the following relationship for the residuals $\|R_i\|_F$ and $\|R_{i-1}\|_F$.
\begin{lemma}
\label{lemma-eco3}
$
\|R_{i}\|_F^2\leq
\|R_{i-1}\|^2_F-\frac{\langle R_{i-1},
\phi_{(u_i,v_i)}\rangle^2}{4}
$
for all $i\geq 1$.
\end{lemma}
\begin{proof} Using the similar argument of Lemma \ref{eco2} and together with the fact that $\langle \phi_{(u_i,v_i)},\phi_{(u_i,v_i)}\rangle =4$, one can easily have this lemma.
\end{proof}

To prove Theorem \ref{mainresult}, we still need several technique lemmas. To state conveniently, let
$$\Phi_E=\{\phi_{(u,v)}: (u,v)\in E\}$$
be the collection of all matrices, we write
$$\langle \Phi_E, \Phi_E\rangle = [ \langle \phi_{(u,v)}, \phi_{(\hat{u},
\hat{v})}\rangle]_{(u,v), (\hat{u}, \hat{v})\in E}$$
to be the Grammian matrix of $\Phi_E$, where
the inner product of two matrices is the standard trace of the product of two matrices. It is easy to see
that the collection of matrices $\phi_{(u,v)}, (u,v)\in E$ are linearly independent and hence,
the Grammian matrix is of full rank and so, the smallest eigenvalue $\lambda_{\min}(\langle \Phi_E, \Phi_E\rangle )>0$. The
following lemma shows that $\lambda_{\min}(\langle \Phi_E, \Phi_E\rangle )\geq 2$ which is essential in our argument, and we believe that it is of independent interest.

\begin{lemma}
\label{laixie12212019}
Let $\langle \Phi_E, \Phi_E\rangle $ be the Grammian matrix of $\Phi_E$. Then the smallest eigenvalue of $\langle \Phi_E, \Phi_E\rangle$ is at least $2$.
\end{lemma}
\begin{proof}
We reshape the matrices $\phi_{(u,v)}$ in to vectors $m_{(u,v)}$, let
$$
M=[\cdots, m_{(u,v)},
\cdots]\big|_{(u,v)\in E}
$$
be an $|V|^2\times |E|$ matrix formed by all reshaped basis vectors.
For convenience, we enumerate vertices by indices $u=1, \ldots, n$ and $(u,v)\in E$
is the same as the $u^{th}$ vertex and $v^{th}$ vertex has an edge in $E$. By the
definition of $\langle \Phi_E, \Phi_E\rangle $, we know that
$$
\langle \Phi_E, \Phi_E\rangle =M^\top M.
$$
Let $M(i,j)$ denote the $(i,j)$-th component of $M$, then
$$
M(i,j)=\left\{
\begin{array}{ll}
-1,  \ \ \ \ \ \mbox{if} \ \ i=(u-1)n+v \ \ \mbox{and} \ \ j=\sum\limits_{i=1}^{v-1}(n-i)+u-v,
\\
-1,  \ \ \ \ \ \mbox{if} \ \ i=(v-1)n+u \ \ \mbox{and} \ \ j=\sum\limits_{i=1}^{v-1}(n-i)+u-v,
\\
1,  \ \ \ \ \ \ \ \mbox{if} \ \ i=(u-1)n+u \ \ \mbox{and} \ \ j=\sum\limits_{i=1}^{v-1}(n-i)+u-v,
\\
1, \ \ \ \ \ \ \ \mbox{if} \ \ i=(v-1)n+v \ \ \mbox{and} \ \ j=\sum\limits_{i=1}^{v-1}(n-i)+u-v,
\\
0, \ \ \ \ \ \ \ \mbox{otherwise}.
\end{array}
\right.
$$
So for any ${\bf x}= \{x_{(u,v)}, (u,v)\in E\}$, we have
$$
(M {\bf x})_{i}=\left\{
\begin{array}{ll}
-x_{(u,v)}, \ \ \ \ \ \ \ \ \ \ \ \ \ \ \ \ \ \  \mbox{if} \ \ i=(u-1)n+v \ \ \mbox{or} \ \ i=(v-1)n+u,
\\
\sum\limits_{ u,v} (x_{(u,j)} +x_{(j,v)}), \ \ \ \ \ \mbox{if} \ \ i=(j-1)n+j,
\end{array}
\right.
$$
where $(M{\bf x})_{i}$ denotes the $i$-th component of $M\bf x$. Therefore,
$$
{\bf x}^\top \langle \Phi_E, \Phi_E\rangle {\bf x}=\|M{\bf x}\|^2_2
= 2\sum_{(u,v)\in E} x_{(u,v)}^2 + \sum_{j=1}^n\bigg(
\sum_{ u,v}
\big(x_{(u,j)}+x_{(j,v)}\big) \bigg)^2
\ge 2 \|{\bf x}\|_2^2.
$$
Hence, the smallest eigenvalue $\lambda_{\min}(\langle \Phi_E, \Phi_E\rangle )\geq 2$.
\end{proof}

We also need the following lemma.
For convenience, let $\tau= \|L_G- L_\epsilon\|^2_F$.
\begin{lemma}
\label{mjlai12262019}
If $\|R_{i}\|^2_F>\tau$, then
\begin{equation}
\label{monotone}
\|R_{i}\|_F^2 -\tau \le \bigg(1- \frac{1}{2|S_{\epsilon}\cup E_{i-1}|}\bigg)(\|R_{i-1}\|_F^2- \tau).
\end{equation}
\end{lemma}
\begin{proof}
To state conveniently, define $C_{\epsilon}(u,v)=0$ for $(u,v)\in E\setminus S_{\epsilon}$, then $L_{\epsilon}=\Phi_{E} C_{\epsilon}$. Similarly, $L_{H_{i-1}}=  \Phi_E C_{i-1}$ with vector $C_{i-1}\in \mathbb{R}^{|E|}$ whose zero entries is in $E\setminus E_{i-1}$. It follows from Lemma~\ref{lemma-eco1} that $L_{H_{i-1}}$ is orthogonal to $R_{i-1}$, so we have
\begin{equation}
\label{key3n}
\begin{array}{ll}
|\langle L_\epsilon, R_{i-1}\rangle| &= |\langle L_\epsilon- L_{H_{i-1}},R_{i-1}\rangle| \\[0.2cm]
&= |\sum_{(u,v)\in E} (C_{i-1}(u,v) - C_\epsilon(u,v))
\langle \phi_{(u,v)}, R_{i-1}\rangle| \\[0.3cm]
&\le \|C_\epsilon -C_{i-1}\|_1 \cdot\max_{(u,v)\in E}
|\langle \phi_{(u,v)}, R_{i-1} \rangle| \\[0.3cm]
&\le \sqrt{|S_{\epsilon} \cup E_{i-1}|}
\|C_\epsilon -C_{i-1}\|_2 \cdot|\langle R_{i-1}, \phi_{(u_{i},v_{i})}\rangle|,
\end{array}
\end{equation}
and by Lemma~\ref{laixie12212019}, we have
\begin{equation}
\label{key4n}
\|L_\epsilon- L_{H_{i-1}}\|_F^2 = \|\Phi_E (C_\epsilon -C_{i-1})\|_F^2\ge
\|C_\epsilon -C_{i-1}\|_2^2 \lambda_{\min}(\langle \Phi_E, \Phi_E\rangle)
\ge 2 \|C_\epsilon -C_{i-1}\|_2^2.
\end{equation}
 We claim that
\begin{equation}
\label{key5n}
\|R_{i-1}\|_F^2 - \tau =
\|R_{i-1}\|_F^2- \|L_G - L_\epsilon\|_F^2 \le \frac{
|\langle L_\epsilon,R_{i-1}\rangle|^2}{\|L_\epsilon- L_{H_{i-1}}\|_F^2}.
\end{equation}
Indeed, using the inequality $ab\leq \frac{1}{2}(a^2+b^2)$ and Lemma~\ref{lemma-eco1} we have
\begin{align*}
\|L_{H_{i-1}} - L_\epsilon\|_F &\sqrt{\|R_{i-1}\|_F^2- \tau^2}
\le \frac{1}{2}(\|L_{H_{i-1}} - L_\epsilon\|_F^2
+ \|R_{i-1}\|_F^2- \|L_G - L_\epsilon\|_F^2)\cr
&= \frac{1}{2}(\|L_{H_{i-1}} - L_\epsilon\|_F^2 + \|R_{i-1}\|_F^2- \|L_{H_{i-1}}- L_\epsilon
+R_{i-1}\|_F^2)\cr
&\le |\langle L_{H_{i-1}} -L_\epsilon, R_{i-1}\rangle|=
|\langle L_\epsilon, R_{i-1}\rangle|.
\end{align*}
 Combining (\ref{key3n}) and
(\ref{key4n}), (\ref{key5n}) can be estimated as
$$
\|R_{i-1}\|_F^2 - \tau   \le \frac{
|\langle L_\epsilon,R_{i-1}\rangle|^2}{\|L_\epsilon- L_{H_{i-1}}\|_F^2}
\le \frac{|S_{\epsilon} \cup E_{i-1}| |\langle R_{i-1}, \phi_{(u_{i},v_{i})}\rangle|^2}
{2}.
$$

On the other hand, by Lemma~\ref{lemma-eco3} we have
$$
\|R_{i}\|_F^2 \le \|R_{i-1}\|_F^2 -|\langle R_{i-1}, \phi_{(u_{i},v_{i})}\rangle|^2/4
\label{key1n}
$$
or
\begin{align*}
\|R_{i}\|_F^2-\tau &\le \|R_{i-1}\|_F^2 -\tau
-|\langle R_{i-1}, \phi_{(u_{i},v_{i})}\rangle|^2/4\cr
&\le \|R_{i-1}\|_F^2 -\tau  - \frac{1}{2|S_{\epsilon} \cup E_{i-1}|} (\|R_{i-1}\|_F^2 - \tau)\cr
&= \bigg(1- \frac{1}{2|S_{\epsilon} \cup E_{i-1}|}\bigg)(\|R_{i-1}\|_F^2 - \tau)
\end{align*}
where we have used (\ref{key5n}).
We have thus obtained the desired result.
\end{proof}

\begin{proof}[Proof of Theorem~\ref{mainresult}]
If $\|R_{i}\|^2_{F}>\tau$, let $\delta=\min\bigg\{\sqrt{1- \frac{1}{2|E|}},\sqrt{1-\frac{\epsilon^2}{4n}}\bigg\}$, then by  Lemma~\ref{mjlai12262019} we have
$$
\|R_{i}\|_F^2-\tau\leq \delta^{2i}(\|L_G\|_F^2-\tau)\leq \delta^{2i}\|L_G\|_F^2.
$$
which togethers with $\|R_{i}\|_{F}-\sqrt{\tau}\leq\sqrt{\|R_{i}\|^2_{F}-\tau}$ implies
$$
\|R_{i}\|_{F}\leq \sqrt{\tau}+\delta^{i}\|L_G\|_F=\|L_G-L_{\epsilon}\|_F+\delta^{i}\|L_G\|_F.
$$
Hence
$$
\begin{array}{ll}
\|L_{H_i}-L_{\epsilon}\|_F&\leq \|L_G-L_{H_i}\|_F+\|L_G-L_{\epsilon}\|_F
\\
&=\|R_{i}\|_F+\|L_G-L_{\epsilon}\|_F
\\
&\leq 2\|L_G-L_{\epsilon}\|_F+\delta^{i}\|L_G\|_F.
\end{array}
$$
This completes the proof.
\end{proof}

\subsection{Spectral Sparsification}
\label{sect.srg}
This subsection aims to discuss that the sparsifier output by  Algorithm~\ref{alg-eco} is a good spectral sparsification. It is well known that the Laplacian matrix $L_G$ has an eigenvalue $0$ and the eigenvalue $0$
will be a multiple linearly independent eigenvalue vectors $\bfu_1, \cdots, \bfu_k$ if $G$ has $k$ clusters. Let us write
$\lambda_1\le \cdots \le \lambda_n$ to be the eigenvalues of $L_G$ and $\lambda_{k+1}>0$ be the first nonzero eigenvalue. Then we have
 \begin{theorem}
\label{The2021-2-5}
Suppose that the graph $G$ has $k$ clusters so that the eigenvalue $\lambda_{k+1}>0$ is the first nonzero eigenvalue of $L_G$. Let $R_i$ be the residual matrix defined in \eqref{residual:eco} of Algorithm~\ref{alg-eco} for  spectral sparsification. Then
\begin{equation}
\label{LGLH}
\bigg(1- \frac{\|R_{i}\|_2}{\lambda_{k+1}}\bigg) L_G \preceq  L_{H_{i}}  \preceq \bigg(1+\frac{\|R_{i}\|_2}{\lambda_{k+1}}\bigg)  L_G
\end{equation}
if $\frac{\|R_{i}\|_2}{\lambda_{k+1}}<1$.
\end{theorem}
 To prove this result, we need a preparation result below.

\begin{lemma}
\label{xu5182020}
Assume that $H$ is a subgraph of $G$. Then the null space of $L_H$ contains the null space of $L_G$.
\end{lemma}
\begin{proof}
It is enough to show that $L_H \bfx = 0$ if  $L_G\bfx = 0$. Suppose  that $L_G\bfx = 0$ for nonzero vector $\bfx$.
Then $\bfx^\top L_G \bfx = 0$ which implies that
$$
\bfx^\top (\bfe_u - \bfe_v )(\bfe_u - \bfe_v )^\top \bfx =0
$$
for all $(u,v)\in E$. It follows that $(\bfe_u - \bfe_v)^\top \bfx =0$.  Hence, for $E'\subset E$, we have
$$
L_H \bfx= \sum_{(u,v) \in E'} c_{(u,v)} (\bfe_u - \bfe_v) (\bfe_u- \bfe_v)^\top \bfx =0.
$$
This completes the proof. \end{proof}

 Now we are ready to prove Theorem \ref{The2021-2-5}.
 \begin{proof}[Proof of Theorem \ref{The2021-2-5}]
Let $\mathcal{N}(L_G)$ be the null space of $L_G$.  For any $\bfx\in \mathcal{N}(L_G)$, we have
$$
(1- \|R_{i}\|_2/\lambda_{k+1})\bfx^\top L_G \bfx  =\bfx^\top L_{H_i}\bfx = (1+\|R_{i}\|_2/\lambda_{k+1}) \bfx^\top L_G\bfx =0
$$
by using Lemma~\ref{xu5182020}. Next  we know
$$
\lambda_{k+1}= \min_{\bfx\in \mathbb{R}^n \backslash \mathcal{N}(L_G)} \frac{\bfx^\top L_G \bfx}{\bfx^\top\bfx} \le
\frac{\bfx^\top L_G \bfx}{\bfx^\top\bfx}
$$
for any $\bfx\not\in \mathcal{N}(L_G)$. In other words,
for such vector $\bfx$, we have $\bfx^\top\bfx \le \bfx^\top L_G \bfx/\lambda_{k+1}$.  Hence, for any $\bfx\not\in \mathcal{N}(L_G)$,
$$
\begin{array}{ll}
\bfx^\top L_{H_{i}}\bfx &= \bfx^\top L_G\bfx + \bfx^\top (L_{H_{i}}-L_G) \bfx
\\
&=\bfx^\top L_G\bfx - \bfx^\top R_{i} \bfx
\\
&\geq \bfx^\top L_G\bfx
-\|R_{i}\|_2 \bfx^\top \bfx
\\
&\geq  \bfx^\top L_G\bfx - \|R_{i}\|_2\bfx^\top L_G\bfx/\lambda_{k+1}
\\
&\geq (1-\|R_{i}\|_2/\lambda_{k+1})\bfx^\top L_G\bfx.
\end{array}
$$
if $\frac{\|R_{i}\|_2}{\lambda_{k+1}}<1$.
That is, we have the left-hand side of (\ref{LGLH}). Similarly, we have the right-hand side of (\ref{LGLH}). These complete the proof.
\end{proof}

By Theorem \ref{The2021-2-5}, we know that if
\begin{equation}\label{desirediq}
\frac{\|R_{i}\|_2}{\lambda_{k+1}}\leq O(\epsilon),
\end{equation}
then we find the desired $\frac{1-O(\epsilon)}{1+O(\epsilon)}$--spectral sparsifier. Next, we will establish \eqref{desirediq}.
Let us first establish another convergence result.
\begin{theorem}
\label{the-eco2}
Suppose that $R_i$ is the residual matrix defined in \eqref{residual:eco} of Algorithm~\ref{alg-eco}. Then
$$\|R_{i+1}\|_F^2 \leq \frac{\|L_G\|_F^2}{1+i\|L_G\|_F^2/(\mbox{Tr}(L_G))^2}, \ \ \forall i\geq1,$$
where $\mbox{Tr}(L_G)$ stands for the trace of matrix $L_G$.
\end{theorem}
\begin{proof}
First by Lemma~\ref{lemma-eco1},
\begin{equation}
\begin{array}{ll}
 \|R_{i-1}\|_F^2&=\langle R_{i-1}, R_{i-1}\rangle= \langle R_{i-1}, L_G\rangle \\
&=\sum\limits_{(u,v)\in E} w_{(u,v)}\langle R_{i-1}, \phi_{(u,v)}\rangle \\
&\leq \|{\bf w}\|_1  |\langle R_{i-1}, \phi_{(u_{i},v_{i})}\rangle|
=\frac{\mbox{Tr}(L_G)}{2}  |\langle R_{i-1}, \phi_{(u_{i},v_{i})}\rangle|,
\end{array}
\end{equation}
 where the inequality follows from Step $2$ of Algorithm~\ref{alg-eco}. It now follows from Lemma~\ref{lemma-eco3} that
$$
\|R_{i}\|_F^2 \le \|R_{i-1}\|_F^2 - \frac{\|R_{i-1}\|_F^4}{(\mbox{Tr}(L_G))^2}.
$$
We next use Lemma~\ref{DT96} with $a_k=\|R_k\|_F^2/\|L_G\|_F^2$ and $\beta_k=\|L_G\|_F^2/(\mbox{Tr}(L_G))^2$
to conclude the desired result.
\end{proof}

 \begin{theorem}
\label{mainresult6}
Suppose that the graph $G$ has $k$ clusters so that the eigenvalue $\lambda_{k+1}>0$ is the first nonzero eigenvalue of
$L_G$.  Let $L_{H_i}$ be the iterative matrix defined in \eqref{LH} of Algorithm~\ref{alg-eco}. If the iteration $i
\geq \frac{n}{\epsilon^2}$ and  $\|R_{i+1}\|_2^2\leq O(\frac{1}{n-k})\|R_{i+1}\|_F^2$, then
\begin{equation}
\label{LGLH21}
(1- O(\epsilon)) L_G \preceq  L_{H_{i+1}}  \preceq (1+O(\epsilon))  L_G
\end{equation}
provided that $\frac{\mbox{Tr}(L_G)}{\lambda_{k+1}}\leq O(\sqrt{n(n-k)})$.
\end{theorem}



\begin{proof}
Recall from Theorem \ref{the-eco2}, we have
\begin{equation}
\label{proof-2021125}
\|R_{i+1}\|_F^2 \leq \frac{\|L_G\|_F^2}{1+i\|L_G\|_F^2/(\mbox{Tr}(L_G))^2}\leq \frac{(\mbox{Tr}(L_G))^2}{i}\leq O((n-k)\lambda^2_{k+1}\epsilon^2)
\end{equation}
for $i\ge n/\epsilon^2$.
Hence, $$\frac{\|R_{i+1}\|_2}{\lambda_{k+1}}\leq O(\frac{1}{\sqrt{n-k}})\frac{\|R_{i+1}\|_F}{\lambda_{k+1}}\leq O(\epsilon).$$
This together with Theorem \ref{The2021-2-5} implies this theorem.
\end{proof}

\begin{remark}
The assumption $\frac{\mbox{Tr}(L_G)}{\lambda_{k+1}}\leq O(\sqrt{n(n-k)})$ in
Theorem~\ref{mainresult6} is satisfied for a large
class of graphs. Indeed, since $\frac{\mbox{Tr}(L_G)}{\lambda_{k+1}}= \frac{\lambda_n+\cdots +\lambda_{k+1}}{\lambda_{k+1}}\leq
(n-k)\frac{\lambda_{n}}{\lambda_{k+1}}$,  the assumption requires $\frac{\lambda_n}{\lambda_{k+1}}\leq
O\big(\sqrt{\frac{n}{n-k}}\big)=O(1)$. In general, there are several models of random graph whose eigenvalues of $L_G$ satisfy the
condition $\frac{\lambda_n}{\lambda_{k+1}}=O(1)$ with high probability.
One may refer to, for example \cite{A17,CR11,FK16, LM19}, for more details on graph theory.

For the assumption $\|R_{i+1}\|_2^2\leq O(\frac{1}{n-k})\|R_{i+1}\|_F^2$, we know that if the nonzero eigenvalues of $R_{i+1}$ are
very close to each other, then one can conclude this assumption.
For the random graph, we know that $\frac{\lambda_n}{\lambda_{k+1}}=O(1)$ which implies that the nonzero eigenvalues of $L_G$
are very close.  By (\ref{LGLH}),  nonzero eigenvalues of the output matrix $L_{H_{i+1}}$ are also very close to each other,
then we can expect this assumption. However, one will not be able to prove it for a general matrix $R_{i+1}$.
Without using the assumption, we are not able to establish the estimate in (\ref{LGLH21}). On the other hand side,
as our numerical experiment has shown
that we do have this estimate so far for various random graphs, we leave it to be an open problem to the community and
to our future study.
\end{remark}

\section{Numerical experiments}
The purpose of this section is to demonstrate that both UGA algorithms are  powerful  for graph sparsification
and sparse positive subset selection. It also shows that the sparse approximation produced by the algorithms
have the potential to improve the performance of graph clustering and least squares regression. All
 experiments were performed on a PC with the processor Intel(R) Core(TM) I5-1035G4 CPU @ 1.50GHz and 8GB memory.
 All programs were written in {\sc Matlab} R2019b.

\subsection{Performance of the UGA algorithms}
The main empirical question is whether the UGA algorithms can find the desired sparifiers correctly.
The following examples demonstrate that UGA algorithm for graph sparsifiers has good performance.
Note that the algorithms stop if the iteration $i>\lceil\frac{n}{\epsilon^2}\rceil.$
This ensures that the output sparse approximation has at most $\lceil\frac{n}{\epsilon^2}\rceil$
edges or vectors. So if the
output matrix $L$ by Algorithm \ref{alg-ecoIso} or  $L_{H}$ by Algorithm \ref{alg-eco} satisfies
\begin{equation}
\label{successrule2}
(1-\epsilon)^2I_n\preceq L\preceq(1+\epsilon)^2I_n
\end{equation}
or
\begin{equation}
\label{successrule}
(1-\epsilon)^2L_{G}\preceq L_{H}\preceq(1+\epsilon)^2L_G,
\end{equation}
then we  conclude that UGA finds the desired sparse approximation successfully.

\begin{example}
We first report the performance of  Algorithm \ref{alg-eco}. The experiment is described as follows.
We first generate random graph $G=(V,E,{\bf w})$ with weights  generated from $i.i.d.$ Poisson distribution or
exponential distribution with parameter (or expectation) $\lambda$. We experiment with $\lambda=1$ and $10$.
Parameter $\epsilon$ is varied from $0.2$ to $0.95$ with the step size $0.05$ as well as $\epsilon=0.99$.
For each $\epsilon$, we repeat $100$ times and calculate the averaged successful rates.
A trial is successful when the output Laplacian matrix $L_{H}$ satisfies \eqref{successrule}.
The results for $\lambda=1$ and $\lambda=10$ are summarized in Figure \ref{example1}, where one can see that
Algorithm \ref{alg-eco} can achieve a successful rate of $100\%$ when $\epsilon\leq 0.55$.
\begin{figure}[hptb]
 \centering
 \begin{tabular}{cc}
 \subfigure[Poisson distribution with $\lambda=1$]{\includegraphics[width=0.46\linewidth]{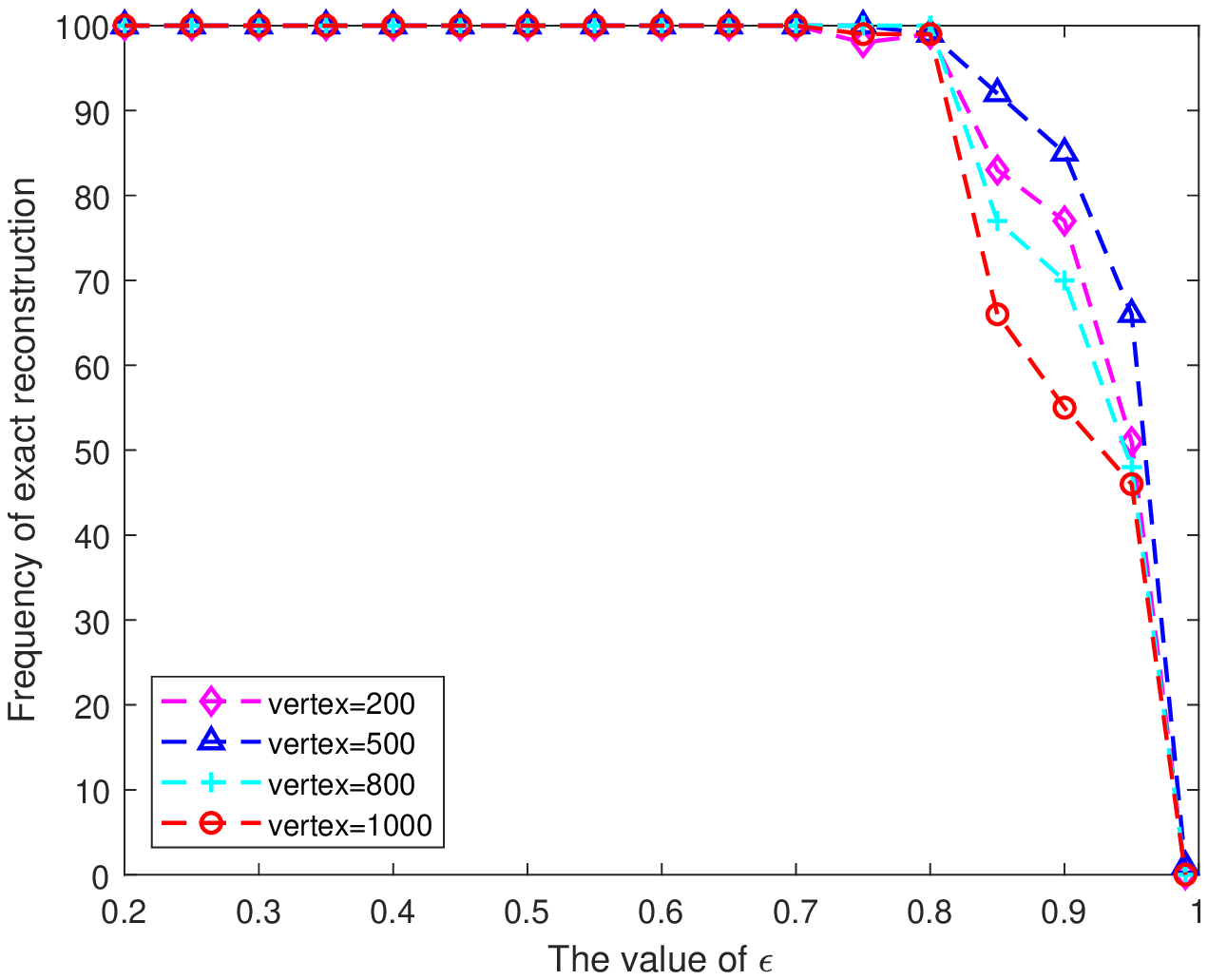}}
 \subfigure[Poisson distribution with $\lambda=10$]{\includegraphics[width=0.46\linewidth]{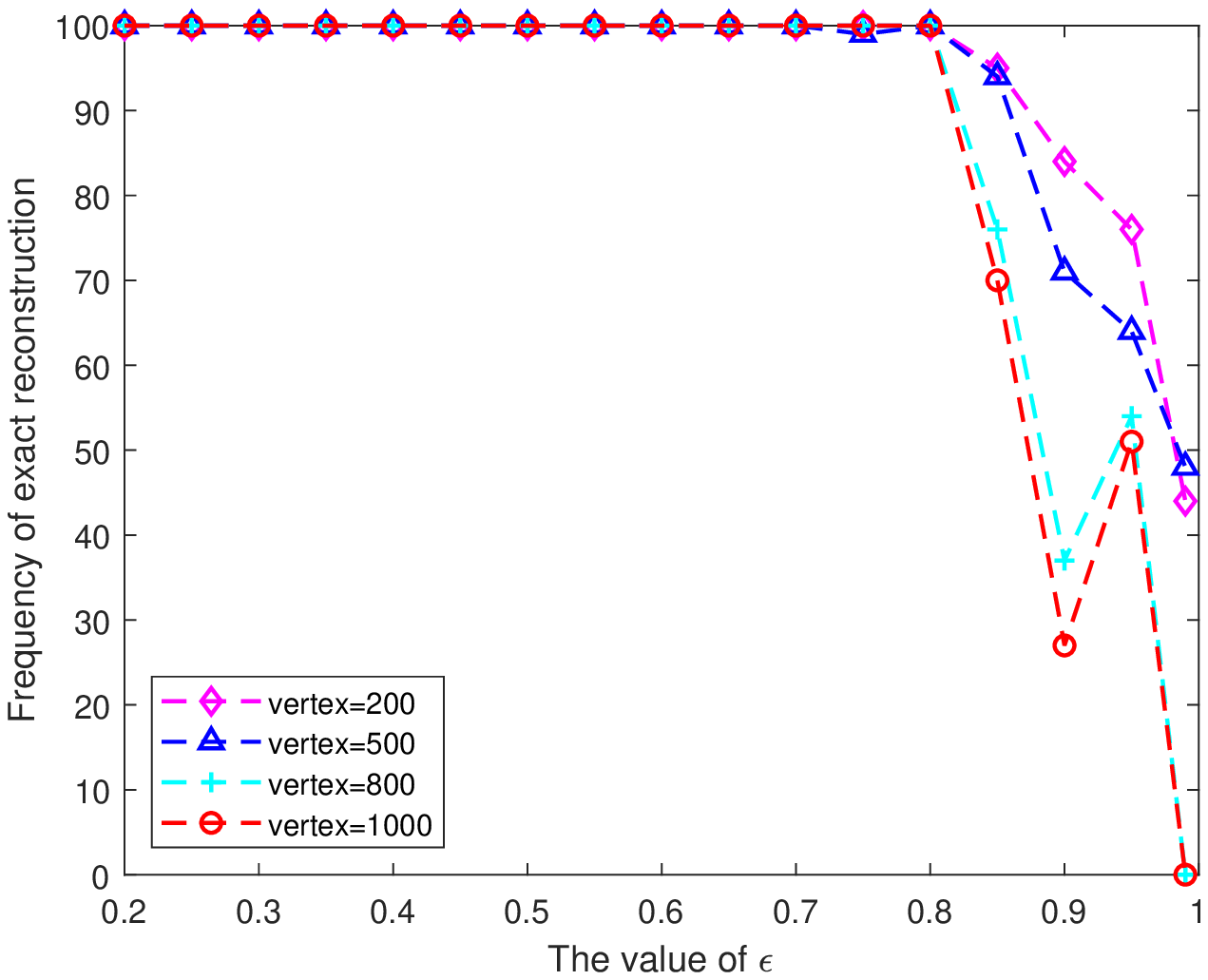}}
 \\
  \subfigure[Exponential distribution with $\lambda=1$]{\includegraphics[width=0.46\linewidth]{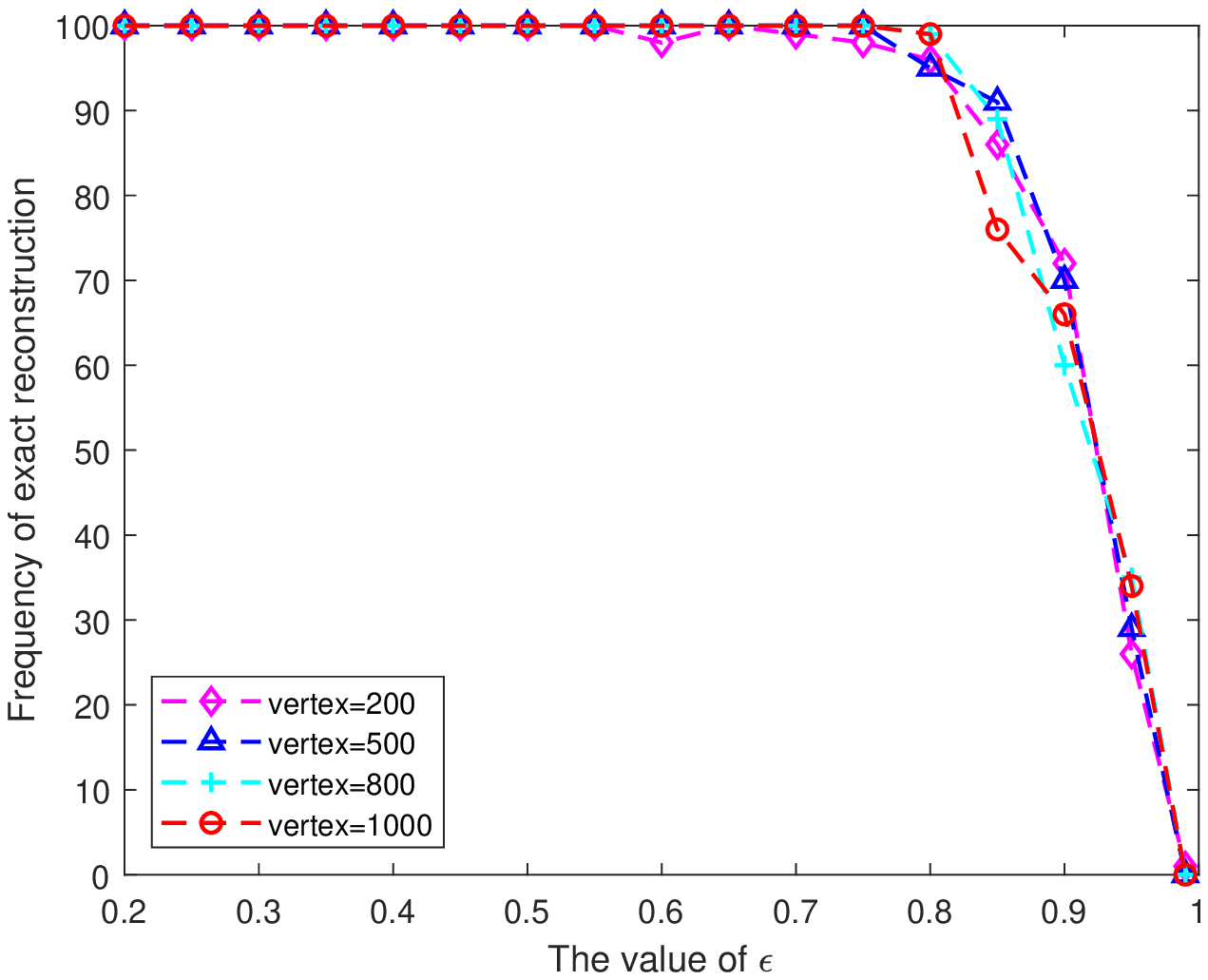}}
 \subfigure[Exponential distribution with $\lambda=10$]{\includegraphics[width=0.46\linewidth]{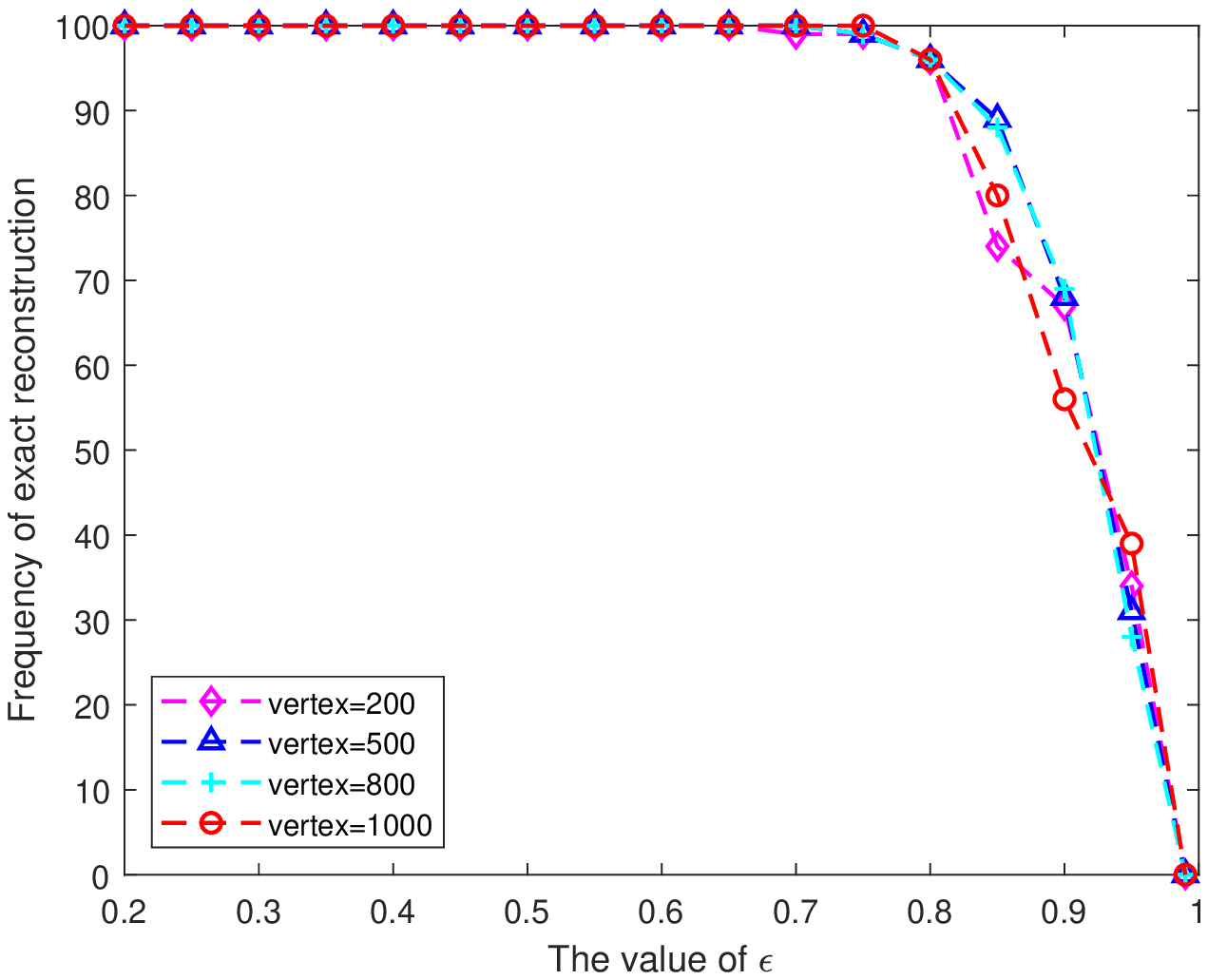}}
 \end{tabular}
 \caption{Success rate experiments. (1) Poisson distribution; (2) Exponential distribution.}
 \label{example1}
 \end{figure}
\end{example}

\begin{example}
This example consists of experiments on graphs sparsification for graphs subject to stochastic block model (SBM) \cite{A17, LM19}.
That is, all the vertices of a graph $G$
are divided into a few clusters, say $k$ clusters, $C_1, \ldots, C_k$ and there is  a
probability matrix $P=[p_{ij}]_{1\le i, j\le k}$ associated with $G$ with $p_{ij}\in
[0,1]$ such that possible edges among vertices in $C_i$ are subject to $p_{ii}$,
$i=1, \ldots, k$ and possible edge between a vertex in $C_i$ and a vertex in $C_j$ is
subject to the probability $p_{ij}$. Since there are more edges within a cluster than
among clusters, we have $p_{ij}<\min\{p_{ii}, p_{jj}\}$ for all $i\not=j$.
We run the experiments with $p_{ii}=p=0.1,i=1, \ldots, k$ and
$p_{i,j}=q=0.01,i\not=j$ for numbers of vertices $=500, 1000, 1500$ with $k$ clusters for $k=2,4,6$.
The results are summarized in Figure \ref{SBMgraph1}, where a trial is successful when the
output graph Laplacian  $L_{H}$ satisfies \eqref{successrule}.
 \begin{figure}[hptb]
 \centering
 \begin{tabular}{cc}
 \includegraphics[width=0.48\linewidth]{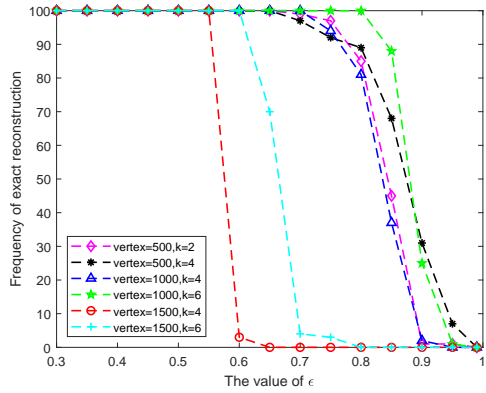}
 \end{tabular}
 \caption{Success rate experiments for SBM. }
 \label{SBMgraph1}
 \end{figure}

The figure shows that for $\epsilon\leq 0.55$, Algorithm \ref{alg-eco} can exactly find the desired sparifiers.
It is also interesting to see that for an fixed vertex number, the larger the cluster $k$, the higher the successful rate of Algorithm \ref{alg-eco}.
\end{example}

\begin{example}
We now show the performance of Algorithm \ref{alg-ecoIso}. We randomly generate the isotropic sets
$V=\{\bfv_1,\ldots,\bfv_m\}\subset\mathbb{R}^n$, where $V=Q^T\in\mathbb{R}^{n\times m}$ with
$Q$ being obtained by QR factorization on a $m\times n$ Gaussian matrix. So we have
$VV^\top=\sum\limits_{i=1}^n\bfv_i\bfv_i^\top =I_n$. The first test was done for fixed $n=100$ with different
 $m=n^2, 2n^2, 4n^2, 8n^2, 16n^2$. In the second test, we set $m=n^2$ for $n=100, 200, \cdots, 500$.
For each pair $(n,m)$, we repeat the experiments for 100 trials and calculate the averaged successful rate based
on \eqref{successrule2}.
Figure \ref{example2} shows the numerical results. The figure shows Algorithm \ref{alg-ecoIso}  performs very well.
In addition, it is interesting to note that the probability of success is descending first and then ascending with respect to the
parameter $\epsilon$. For example, it can be see that the probability of success with $\epsilon=0.99$ is higher than that
with $\epsilon=0.95$ which can be explained as follows. From the stopping criterion \eqref{successrule2}, we know that
when the parameter $\epsilon$ is  closing to $1$, there is a dramatic increment of the ratio of $(1+\epsilon)^2/(1-\epsilon)^2$,
but only a slight decrement of the size of the selected set by Algorithm \ref{alg-ecoIso}.
 \begin{figure}[hptb]
 \centering
 \begin{tabular}{cc}
 \subfigure[$n=100$]{\includegraphics[width=0.46\linewidth]{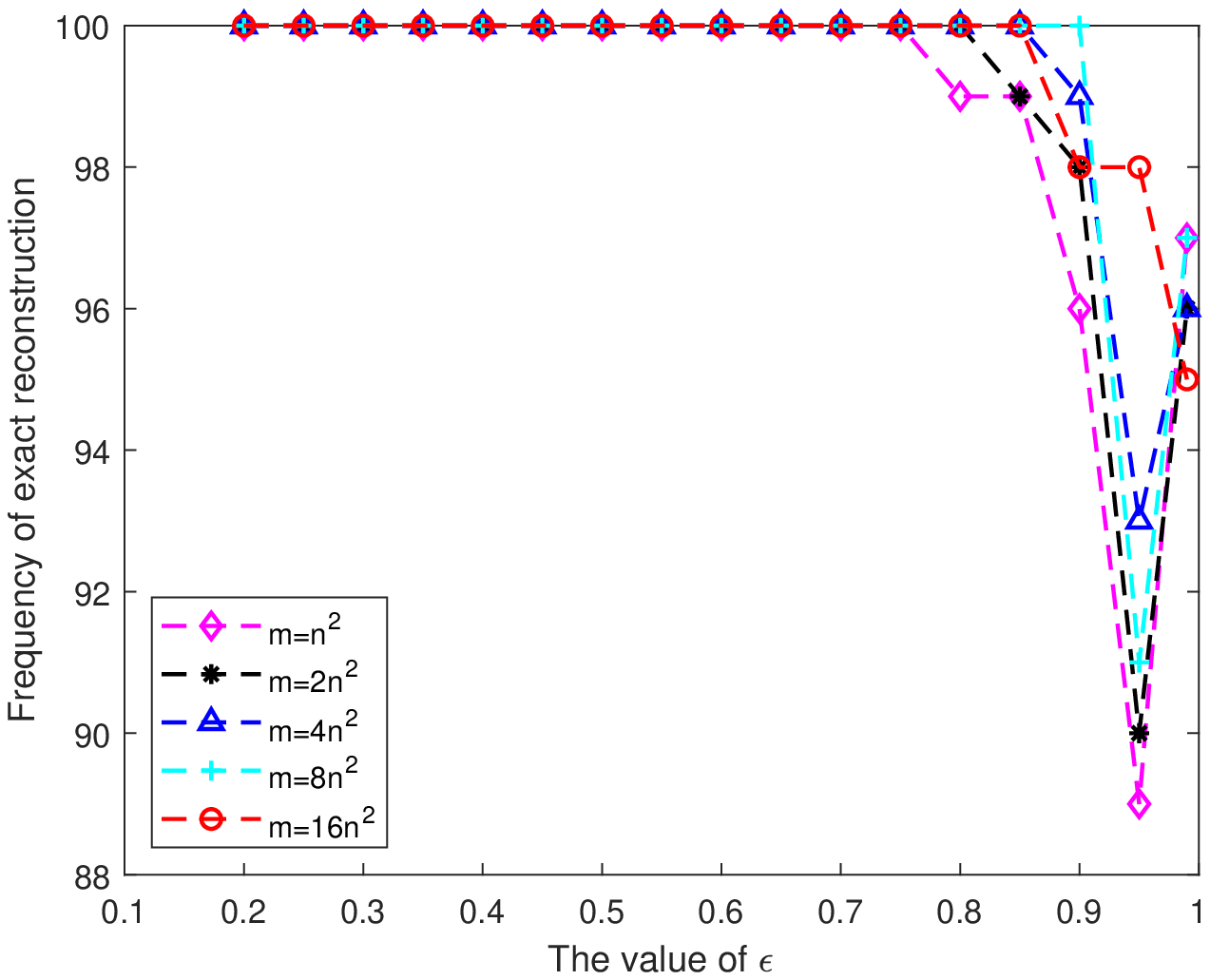}}
 \subfigure[$m=n^2$]{\includegraphics[width=0.46\linewidth]{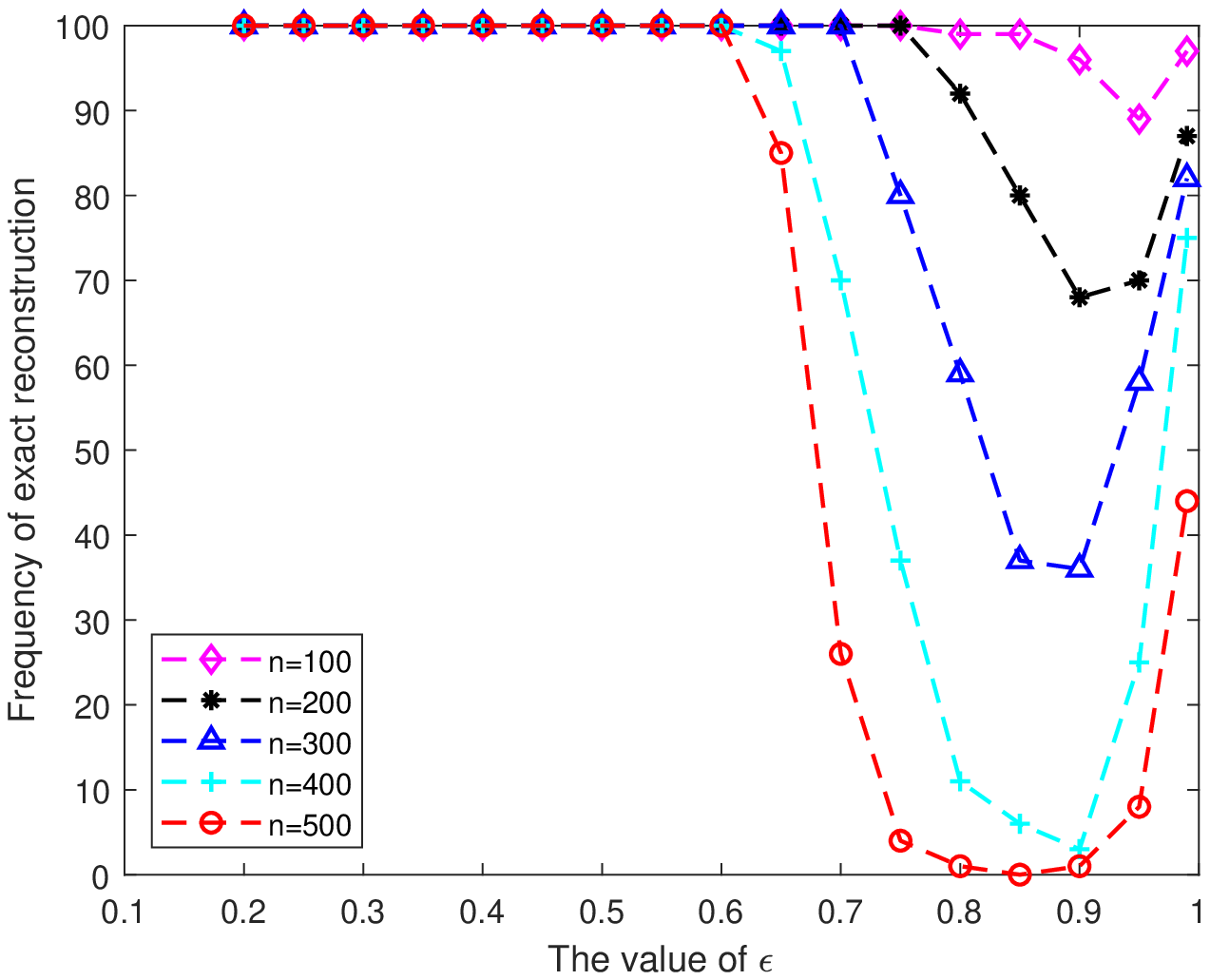}}
 \end{tabular}
 \caption{Successful rate experiments. (a)  We take fixed $n=100$ and change $m=n^2,2n^2,4n^2,8n^2,16n^2$;
 (b) We set $m=n^2$ for $n=100,200,300,400,500$. }
 \label{example2}
 \end{figure}
\end{example}

\subsection{Graph sparsification for graph clustering}
This subsection aims to show that the sparsifiers produced by Algorithm \ref{alg-eco}
improve the performance of graph clustering.
Graph clustering is one of the major research activities in graph data analysis. It has
many applications in communities detection, social network analysis, and etc. Mainly, given
a weighted or unweighted adjacency matrix $A$ associated with a graph $G=(V,E)$, one needs to
find clusters, i.e. $V=C_1 \cup \cdots \cup C_k$ so that the adjacency matrix $A$ permuted according
to the order of the indices in $C_1$, then $C_2, \cdots, C_k$ is an almost blockly diagonal
matrix.   In following examples, we demonstrate that the computation of graph clustering
can be simplified in time of computation and the accuracy of the clustering can be improved.
We shall use a standard graph model called stochastic block model (SMB)(cf. \cite{A17,LM19})  to generate graphs.

\begin{example}
\label{firstex}

 \begin{figure}[hptb]
 \centering
 \begin{tabular}{cc}
\includegraphics[width=0.6\linewidth]{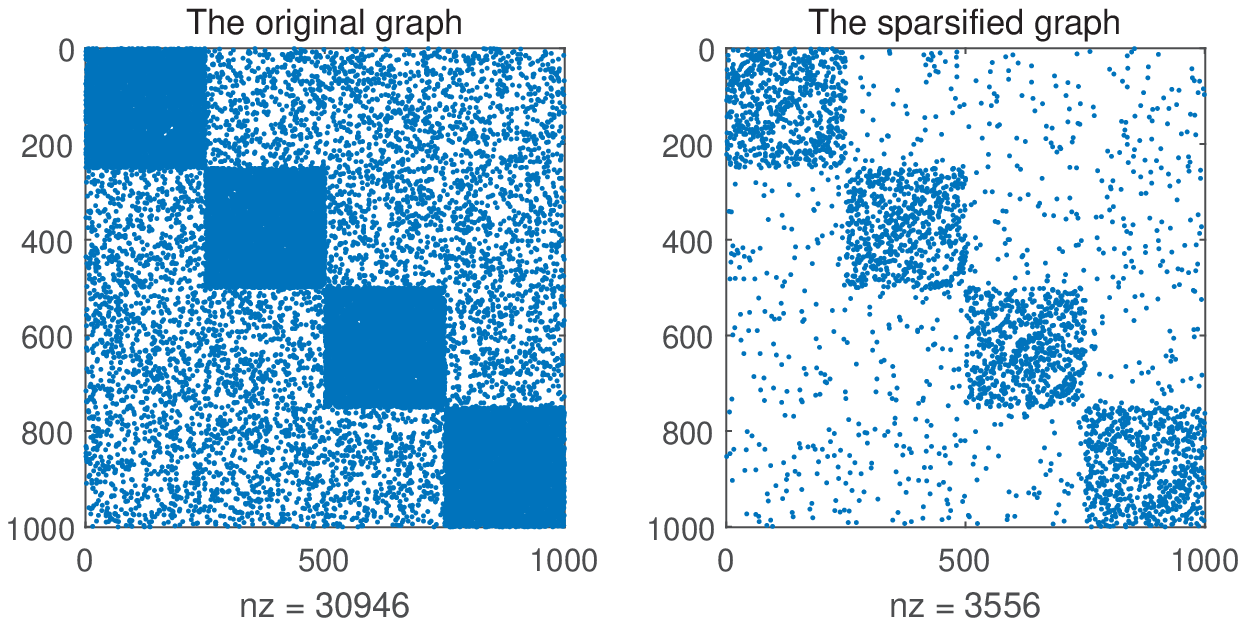} \\
\includegraphics[width=0.6\linewidth]{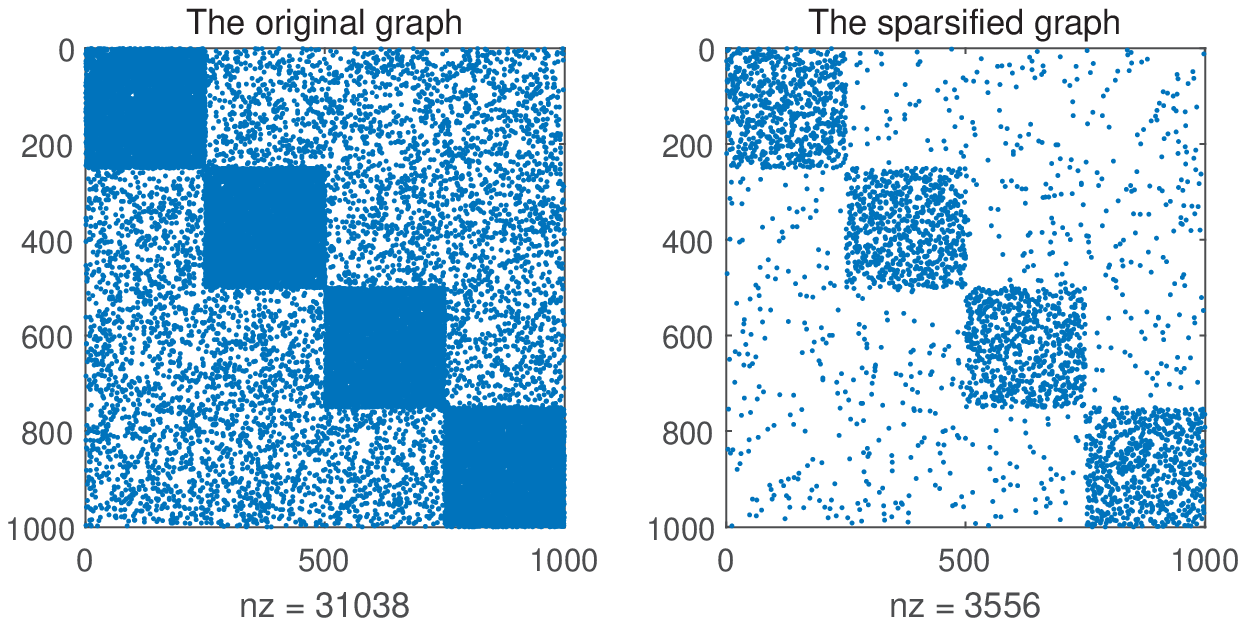}
 \end{tabular}
 \caption{Adjacency matrices of an original and sparsifier graph: two experiments. }
 \label{example4}
 \end{figure}

We have run the experiments with $p=0.1$ and $q=0.008$ for various number $n>1$ of vertices,
e.g. $n=1000$ with $k$ clusters for various $k>1$, e.g. $k=4$.
In Figure~\ref{example4}, we present two adjacency matrices together with their sparsifiers obtained from UGA
for graph sparsifiers with $\epsilon=0.75$. Two experiments are shown with $k=4$ and $n=1000$
that the structure of clusters is well preserved.
\end{example}

\begin{example}
\label{ex1b}
Next we consider random adjacency matrices and  then use our graph sparsification algorithm. Again
we use the stochastic block model to generate graphs as in Example~\ref{firstex}.
After permuting the columns and rows of the adjacency
matrix $A$ randomly,  we apply a standard spectral clustering method \cite{NJW01} to find
clusters with known number of clusters. Also, we first apply the UGA algorithm for graph
sparsifiers and then apply the spectral clustering method to find clusters. We demonstrate that
the accuracy of clusters obtained from sparsified graph can be better than that from the original graph.
Here the accuracy is computed based on the following formula
$$
c_i = \frac{|C_i \cap C_i^\#|}{n_i},
$$
where $n_i$ is the known size of cluster $C_i$ and $C_i^\#$ is the cluster found by the
spectral clustering method. Note in our case, $C_1, \cdots, C_k$ are pre-planted index
sets and thus are known.  We use the average of the accuracies $c_1, \cdots, c_k$ to report
the performance of our computation.

In Figure~\ref{example5}, we use $n=800$ and $p=0.08$ while $q=0.008$, $k=4$ to generate a graph $G$. Then
the spectral clustering method  to find all clusters of $G$ using the original adjacency matrix and sparsified
adjacency matrix by Algorithm~\ref{alg-eco}. We repeat the experiment 100 times and
report the accuracies of clusters for both graphs for each time. The mean of successful rates of finding members
of each cluster correctly from the original graphs is $88.73\%$ while the mean of successful rates from the sparsified graphs
is $93.13\%$. That is, the graph sparsification can improve the accuracy.

 \begin{figure}[hptb]
 \centering
\includegraphics[width=0.6\linewidth]{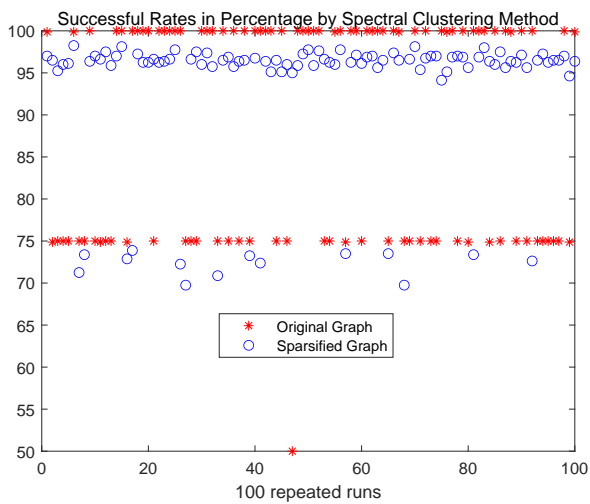}
 \caption{Accuracies in percentage from both graphs over $100$ runs.
 \label{example5}}
 \end{figure}

\end{example}

\subsection{A linear sketching for least squares regression}
 In many applications in statistical data-analysis and  inverse problems,  we need to solve
a  least squares problem.  That is,
we are given an $m\times n$ matrix $A$ and an $m\times 1$ vector $\bfb$ with $m\gg n$. The least
squares regression problem is to  find an  $\bfx_{opt}\in\mathbb{R}^{n}$ such that
\begin{equation}
\label{LS-solution}
\bfx_{opt}\in \arg\min_{\bfx}\|A\bfx- \bfb\|^2_2.
\end{equation}
This subsection aims to show that the sparse approximations produced by Algorithm \ref{alg-ecoIso}
can be useful for  least squares regression.
Firstly, we consider the case that the coefficient matrix $A$ is fixed and $b$ are variable measurements. It has many applications in signal processing and image processing, where $A$ is a fixed signal sensor with many observations $b$.
Indeed, for  a given
$$
 B=  A^\top A = \sum_{i=1}^m  \bfa_i^\top \bfa_i,
$$
where $A=[\bfa_1^\top, \cdots, \bfa_m^\top]^\top\in \mathbb{R}^{m\times n}$ with $m\gg n$ and for any $\epsilon>0$,
we use Algorithm~\ref{alg-ecoIso} to find a sparse solution $d_i\ge 0, i=1, \cdots, s$ with $s=O(n/\epsilon^2)$ such that
\begin{equation}
\label{requirement}
\big(1- \epsilon\big) B \le \sum_{k=1}^s d_k \bfa_{j_k} \bfa_{j_k}^\top \le \big(1+\epsilon\big)  B.
\end{equation}
Then we solve $\tilde{A}\tilde{\bfx} = \tilde{\bfb}$ instead of the original least squares problem, where
$$
\tilde{A}=\left[\begin{matrix} \sqrt{d_1} \bfa_{j_1}\cr \vdots \cr \sqrt{d_s}\bfa_{j_s}\end{matrix}\right]
\hbox{ and }  \tilde{\bfb} = \left[ \begin{matrix} b_{j_1}\cr \vdots\cr b_{j_s}\end{matrix}\right].
$$
Then we claim that $\tilde{\bfx} \approx \bfx$  if
$\|A^\top \bfb - \tilde{A}^\top\tilde{\bfb}\|_2 \approx 0$.
Indeed, we have the following
\begin{theorem}
\label{mjlai12102020}
Let $\bfx$ be the least squares solution satisfying $B\bfx =A^\top \bfb$ and $\tilde{\bfx}$ be the solution $\tilde{A}\tilde{\bfx} =
\tilde{\bfb}$ from the above.  Suppose that $1- \kappa(B) \frac{\|B- \tilde{A}^\top \tilde{A}\|_2}{\|B\|_2}>0$,
where $\kappa(B)$ stands for the condition number of $B$. Then
\begin{equation}
\label{lssapp}
\frac{\|\bfx - \tilde{\bfx}\|_2}{\|\bfx \|_2} \le  \frac{\kappa(B)}{1- \kappa(B)
\frac{\|B- \tilde{A}^\top \tilde{A}\|_2}{\|B\|_2}}
\left( \frac{\|B- \tilde{A}^\top \tilde{A}\|_2}{\|B\|_2}
+ \frac{\|A^\top \bfb - \tilde{A}^\top\tilde{\bfb}\|_2}{\|A^\top \bfb\|_2}\right).
\end{equation}
\end{theorem}
\begin{proof}
We mainly use the stability of linear system $B\bfx =A^\top A \bfx=A^\top \bfb$ to establish (\ref{lssapp}).
\end{proof}

By (\ref{requirement}), we have
$\|B- \tilde{A}^\top \tilde{A}\|_2= \max\limits_{\|\bfx\|_2=1} \bfx^\top (B-\tilde{A}^\top \tilde{A})\bfx
\le O(\epsilon \|B\|_2).$  If
$\|A^\top \bfb - \tilde{A}^\top\tilde{\bfb}\|_2  \to 0$, then the right-hand side of
(\ref{lssapp}) is very small, and hence, the approximating solution $\tilde{\bfx}$ is close to the true solution $\bfx$.

Such a computational method is similar to
the linear sketching approach discussed in \cite{boutsidis2013l2, yosef2018, woodruff2014} by
compressing the data $A$ and $\bfb$ as small as possible before doing linear regression.
So our Algorithm~\ref{alg-ecoIso} provides another approach for linear sketching.
Instead of using the linear sketching method to find ${\bf y}= \tilde{A}\backslash \tilde{\bfb}$ we propose
to solve ${\bf z}=(\tilde{A})^\top \tilde{A}\backslash A^\top \bfb$ directly. The reason for this  new solution
can be seen from (\ref{lssapp}) that the solution will make the second term on the right-hand side to be zero.
In the following example, we compare two solutions to demonstrate numerically that our new method is more accurate.

\begin{example}
\label{lsqeg1}
We use Algorithm~\ref{alg-ecoIso} to find a linear sketching $\tilde{A}$ and $\tilde{\bfb}$
with various $\epsilon\in (0,1)$.
Once we find the linear sketching, we repeatedly solve a linear squares regression
by using the linear sketching method: ${\bf y}= \tilde{A}\backslash \tilde{\bfb}$ and our new method
over $100$ times and compute the averaged times and errors.
Let Time-LS, Time-New and Time-O denote the computational times for the two methods
as well as the method using {\sc Matlab} backslash, i.e. ${\bf x}=A\backslash {\bf b}$, respectively.
 We use RE-LS to denote the averaged relative residual error, i.e. $\|{\bf y}-{\bf x}\|_2/\|{\bf x}\|_2$
 by the linear sketching solution method.
 Similarly, RE-New denotes the relative residual error, $\|{\bf z}-{\bf x}\|_2/\|{\bf x}\|_2$.
In Table~\ref{tabn}, we use $n=300,m= 90000$.

\begin{table}[htbp]
\caption{Performance of the linear sketching method and our new method\label{tabn}}
\begin{center}
\begin{tabular} {|c|r|r|r|r|r|r|}
\hline
 $\epsilon$ & RE-LS &RE-New & Time-O &Time-LS &Time-New & BE \\ \hline
    0.90 &   2.7487 &  15.5786 &   2.5359 &   0.0013&   0.0123 &         2 \\
    0.80 &   2.0325 &   4.6479 &   2.7808 &   0.0014&   0.0130 &         2 \\
    0.70 &   1.7395 &   2.1582 &   2.7758 &   0.0014&   0.0127 &         2 \\
    0.60 &   1.5517 &   1.1918 &   2.6825 &   0.0017&   0.0130 &         3 \\
    0.50 &   1.4372 &   0.6701 &   2.6087 &   0.0020&   0.0128 &       6 \\
    0.40 &   1.3603 &   0.3807 &   2.7778 &   0.0023&   0.0132 &       7 \\
    0.30 &   1.2931 &   0.2045 &   2.6791 &   0.0026&   0.0121 &       14 \\
    0.20 &   1.1968 &   0.0917 &   2.7881 &   0.0028&   0.0140 &      29 \\
\hline
\end{tabular}
\end{center}
\end{table}
The numbers in the last column are the numbers of the columns in the right-hand side of the least squares problem to
break even (BE) of the computational time of our new method with {\sc Matlab} least squares method, i.e.
${\bf x}=A\backslash \bfb$. That is, if there are $30$ columns or more on the right-hand side of least squares regression,
we should use Algorithm~\ref{alg-ecoIso} to find $\tilde{A}$ and use our new method to solve them with $\epsilon=0.2$.
\end{example}

\end{document}